\documentclass[10pt]{amsart}
\usepackage{graphicx}
\usepackage{epstopdf}
\usepackage{amsmath}
\usepackage{amsthm}
\usepackage{amsfonts}
\usepackage{amssymb}
\usepackage{hyperref}
\usepackage{enumerate}
\usepackage{bigstrut}

\newtheorem*{Th*}{Theorem}
\newtheorem{proposition}{Proposition}
\newtheorem{corollary}{Corollary}
\newtheorem{theorem}{Theorem}
\newtheorem{lemma}{Lemma}
\theoremstyle{definition}
\newtheorem{remark}{Remark}
\DeclareMathOperator{\arcoth}{arcoth}

\begin{document}

\title{Highly symmetric POVMs and their informational power}

\author{Wojciech S\l omczy\'{n}ski \and  Anna Szymusiak}

\address{Institute of Mathematics, Jagiellonian University, \L ojasiewicza 6,
30-348 Krak\'{o}w, Poland}
\email{Wojciech.Slomczynski@im.uj.edu.pl; Anna.Szymusiak@uj.edu.pl}

\maketitle

\vspace{-0.2cm}

\begin{abstract}
We discuss the dependence of the Shannon entropy of normalized finite
rank-$1$ POVMs on the choice of the input state, looking for the states that
minimize this quantity. To distinguish the class of measurements where the
problem can be solved analytically, we introduce the notion of highly
symmetric POVMs and classify them in dimension two (for qubits). In this case
we prove that the entropy is minimal, and hence the relative entropy (informational power)
is maximal, if and only if the input state is orthogonal to one of the states constituting
a POVM. The method used in the proof, employing the Michel theory of critical
points for group action, the Hermite interpolation and the structure of invariant
polynomials for unitary-antiunitary groups, can also be applied in higher
dimensions and for other entropy-like functions. The links between entropy
minimization and entropic uncertainty relations, the Wehrl entropy and
the quantum dynamical entropy are described. \\

Keywords: POVM, entropy, group action, symmetry, Hermite interpolation

MSC: 81P15, 81R05, 94A17, 94A40, 58D19, 58K05, 52B15

\end{abstract}

\vspace{-0.2cm}

\tableofcontents

\newpage

\section{Introduction}

Uncertainty is an intrinsic property of quantum physics: typically, a
measurement of an observable can yield different results for two identically
prepared states. This indeterminacy can be studied by considering the
probability distribution of measurement outcomes given by the Born rule, and
quantized by a number that characterizes the randomness of this distribution.
The Shannon entropy is the most natural tool for this purpose. Obviously, the
value of this quantity is determined by the choice of the initial state of the
system before the measurement. When the number of possible measurement
outcomes is finite and equals $k$, it varies from $0$, if the measurement
outcome is determined, to $\ln k$, if all outcomes are equiprobable. If the
measured observable is represented by a~normalized rank-$1$ positive-operator valued
measure (POVM) on a $d$-dimensional complex Hilbert space, where $d \leq k$, then the upper
bound is achieved for the maximally mixed state $\mathbb I/d$. On the other hand, the
Shannon entropy of measurement cannot be $0$ unless the POVM is a projection
valued measure (PVM) representing projection (L\"{u}ders-von Neumann)
measurement with $k=d$, since it is bounded from below by $\ln\left(
k/d\right)  $. Thus in the general case the following questions arise: how to
choose the input state to minimize the uncertainty of the measurement
outcomes, and what is the minimum value of the Shannon entropy for the
distribution of measurement results in this case? In the present paper we call
this number the \textsl{entropy of measurement}.

The entropy of measurement has been widely studied by many authors since the
1960s \cite{Weh78}, also in the context of entropic uncertainty principles
\cite{Deu83}, as well as in quantum information theory under the name of
\textsl{minimum output entropy of a quantum\--classical channel} \cite{Sho02}.
Subtracting this quantity from $\ln k$, we get the \textsl{relative entropy of
measurement} (with respect to the uniform distribution), which may vary from
$0$ to $\ln d$. In consequence, the optimization problem now reduces to
finding its maximum value. Either way, we are looking for the `least quantum'
or `most classical' states in the sense that the measurement of the system
prepared in such a state gives the most defined results. The answer is immediate
for a PVM, consisting of projections onto the elements of an orthonormal basis,
that are at the same time `most classical' with respect to this measurement.
Such an obvious solution is not available for general POVM.
Because of concavity of the entropy of measurement as a function of state,
we only know that the optimal states must be pure.

Like many other optimization problems where the Shannon function $\eta\left(
x\right)  =-x\ln x, x>0$ is involved, the minimization of the entropy of
measurement seems to be too difficult to be solved analytically in the general
case. In fact, analytical solutions have been found so far only for a few
two-dimensional (qubit) cases, where the Bloch vectors of POVM elements
constitute an $n$-gon \cite{Sch89,Ghietal03,Arnetal11}, a~tetrahedron
\cite{Oreetal11} or an octahedron \cite{San95,Coletal11}. All these POVMs are
symmetric (group covariant), but, as we shall see, symmetry alone is not
enough to solve the problem analytically. However, for symmetric rank-$1$
POVMs the relative entropy of measurement gains an additional interpretation.
It follows from \cite{Oreetal11}, that it is equal to the
\textsl{informational power of measurement} \cite{Arnetal11,Arnetal14a}, viz.,
the \textsl{classical capacity of a quantum-classical channel} generated by
the POVM \cite{Hol12b}. To distinguish the class of measurements for which the
entropy minimization problem is feasible, we define \textsl{highly symmetric}
(\textsl{HS}) normalized rank-$1$ POVMs as the symmetric subsets of the state space
without non-trivial factors. The primary aim of
this paper is to present a general method of attacking the minimization problem for
such POVMs and to illustrate it, entirely solving the issue in the two-dimensional case.

Note that our method is not confined to qubits, and it works also in higher
dimensions, at least for some important cases, though it is true, that for dimension three
or larger it seems more difficult to be applied, mainly because the image of the Bloch representation
of pure states is only a proper subset of the (generalized) Bloch sphere. However, one of us
(A.S.) has published recently a paper \cite{Szy14}, where the technique developed in an earlier
version of the present paper has been used to find the minimum of the entropy for group
covariant SIC-POVMs in dimension three, including the Hesse HS SIC-POVM. The same method
works also for a POVM consisting of four MUBs, again in dimension three (this result was
first obtained by a different method in \cite{Arn14}),
as well as for the $64$ Hoggar lines HS SIC-POVM in dimension eight \cite{SloSzy15}.
After some additional work, one can prove that the technique developed in this paper for searching
the minima, can be also used to find the maximum entropy of the distribution in question
among pure pre-measurement states for an arbitrary SIC-POVM in any dimension, as well as the maximum entropy
for pure initial states for all HS-POVMs in dimension two \cite{Szy15}. Summarizing, this seems to
be a quite universal technique of finding extrema, limited neither to qubits, nor to the Shannon entropy,
as it can be applied to various `entropy-like' quantities obtained with the help of other functions
with similar properties as $\eta$, such as power functions leading to the R\'{e}nyi entropy or its
variant, the Tsallis-Havrda-Charv\'{a}t entropy \cite{Csi08}, and even to more general `information
functionals' considered in the same context in \cite{Busetal87}.

Going back to the dimension two, we first classify all HS-POVMs, proving that their Bloch
sphere representations must be either one of the five Platonic solids or the
two quasiregular Archimedean solids (the cuboctahedron and icosidodecahedron),
or belong to an infinite series of regular polygons. For such POVMs we show
that their \textit{entropy is minimal (and so the relative entropy is
maximal), if and only if the input state is orthogonal to one of the states constituting
a POVM}. We present a unified proof of this fact for all eight cases, and for
five of them (the cube, icosahedron, dodecahedron, cuboctahedron and
icosidodecahedron) the result seems to be new. Let us emphasize that commonly used
methods of minimizing entropy, e.g. based on majorization cannot be applied
in all these cases.

The proof strategy is as follows. We consider a set $S=\left\{  \sigma
_{j}:j=1,\ldots,k\right\}  $ contained in the space of pure states
(one-dimensional projections) $\mathcal{P}\left(  \mathbb{C}^{d}\right)  $
representing a normalized rank-$1$ HS-POVM. The entropy of measurement $H$ is
given by $H\left(  \rho\right)  =\sum\nolimits_{j=1}^{k}\eta\left(
p_{j}\left(  \rho\right)  \right)  $, where the probabilities of the
measurement outcomes are $p_{j}\left(  \rho\right)  =(d/k)\operatorname{tr}%
\left(  \sigma_{j}\rho\right)  $ for $\rho\in\mathcal{P}\left(  \mathbb{C}%
^{d}\right)  $ and $j=1,\ldots,k$. We start from analysing the group action of
$\operatorname*{Sym}(S)$, the group of unitary-antiunitary symmetries of $S$,
on $\mathcal{P}\left(  \mathbb{C}^{d}\right)  $. We identify points lying in
the maximal stratum for this action, called inert states in physical
literature. As the POVM is highly symmetric, this set contains $S$ itself.
According to the Michel theory of critical orbits of group actions
\cite{Mic71,MicZhi01}, the elements of the maximal stratum, being critical
points for the entropy of measurement $H$, which is a $\operatorname*{Sym}%
(S)$-invariant function, are natural candidates for the minimizers. Studying
their character, we see that $H$ has local minima at the inert states
$\sigma_{j}^{\perp}$ ($j=1,\ldots,k$) orthogonal to the elements of $S$. To
prove that these minima are indeed global we look for a simpler (polynomial)
$\operatorname*{Sym}(S)$-invariant function $P$ such that: $P\leq H$, and
$P=H$ at $\sigma_{j}^{\perp}$ ($j=1,\ldots,k$). To construct such a polynomial
function we define it as $P\left(  \rho\right)  =\sum\nolimits_{j=1}
^{k}p\left(  p_{j}\left(  \rho\right)  \right)  $, for a polynomial $p$ being
a suitable Hermite approximation of $\eta$ at values $p_{j}\left(  \sigma
_{i}^{\perp}\right)  $ ($j=1,\ldots,k$) for some, and hence for all,
$i=1,\ldots,k$. Now it is enough to prove that these `suspicious' points are
global minimizers for $P$, which is an apparently easier task. Proving that
$P$ has minimizers at $\sigma_{j}^{\perp}$ ($j=1,\ldots,k$), we use the fact
that the structure of invariant polynomials for any finite subgroup of the
projective unitary-antiunitary group is well known for $d=2$. Employing a
priori estimates for the degree of $p$, and hence for the degree of $P$, we
can show that $P$ is either constant, which completes the proof, or it is a
low degree polynomial function of known $\operatorname*{Sym}(S)$-invariant
polynomials, which reduces the proof to a relatively easy algebraic problem.
The following two points seem to be crucial to the proof: the form of the
function $\eta$ that guarantees that the Hermite interpolation polynomial $p$
bounds $\eta$ from below, and the knowledge of the subgroups of
unitary-antiunitary group that may act as symmetry groups of the sets
representing HS-POVMs as well as their invariant polynomials.

The problem considered in the present paper has a well-known continuous
counterpart: the minimization of the Wehrl entropy over all pure states, see
Sec.~\ref{Weh}, where the (approximate) quantum measurement is described by an
infinite family of group coherent states generated by a unitary and
irreducible action of a linear group on a highly symmetric fiducial vector
representing the vacuum. More than thirty years ago Lieb \cite{Lie78}, and
quite recently Lieb \& Solovej \cite{LieSol12} proved for harmonic oscillator
and spin coherent states, respectively, that the minimum value of the Wehrl
entropy is attained, when the state before the measurement is also a coherent
state. Surprisingly, an analogous theorem need not be true in the discrete
case, since the entropy of measurement need not be minimal for the states
constituting the POVM. This discrepancy requires further study.

In Sec.~\ref{Eup}\ we show that the minimization of the entropy of measurement
is also closely related to entropic uncertainty principles \cite{WehWin10}.
Indeed, every such principle leads to a lower bound for the entropy of some
measurement, and conversely, such bounds may yield new uncertainty principles
for single or multiple measurements. Moreover, in Sec.~\ref{QDE} we reveal the
connection between the entropy of measurement and the quantum dynamical
entropy with respect to this measurement \cite{SloZyc94}, the quantity
introduced independently by different authors to analyse the results of
consecutive quantum measurements interwind with a given unitary evolution.

The rest of this paper is organized as follows. In Sec.~\ref{QSP} we review
some of the standard material on quantum states and measurements including the
generalized Bloch representation. In Sec.~\ref{SRHS} we analyse the general
notion of highly symmetric sets in metric spaces, and in Sec.~\ref{HSPOVM} we
apply this universal notion to normalized rank-$1$ POVMs. Sec.~\ref{HSPOVM2}
contains the classification of all HS-POVMs in dimension two.
Sec.~\ref{ENTRELENT} provides a detailed exposition of entropy and relative
entropy of quantum measurement as well as their relations to the notions of
informational power and Wehrl entropy, and their connections with entropic
uncertainty principles and quantum symbolic dynamics. In Sec.~\ref{Local} we
study local minima for the entropy of measurement in dimension two, and in
Sec.~\ref{Gloext} we use Hermite interpolation and group invariant polynomials
techniques to derive our main theorem and to find the global minima in this
case. Finally, in Sec.~\ref{INFPOW} we apply the obtained results to give a
formula for the informational power of HS-POVMs in dimension two.

\section{Quantum states and POVMs}
\label{QSP}

In this section we collect all the necessary definitions and facts about
quantum states and measurements that can be found, e.g., in \cite{BenZyc06} or
\cite{HeiZem11}. Consider a quantum system for which the associated complex
Hilbert space $\mathcal{H}$ is finite dimensional, that is $\mathcal{H}%
=\mathbb{C}^{d}$ for some $d=2,3,\ldots$. The \textsl{pure states} of the
system can be described as the elements of the complex projective space
$\mathbb{P}\mathcal{H}=\mathbb{CP}^{d-1}$ endowed with the
\textsl{Fubini-Study} (called also \textsl{procrustean} after Procrustes)
K\"{a}hler metric given by $D_{FS}\left(  \left[  \varphi\right]  ,\left[
\psi\right]  \right)  :=\arccos\frac{\left|  \left\langle \varphi
|\psi\right\rangle \right|  }{\left\|  \varphi\right\|  \left\|  \psi\right\|
}$ for $\varphi,\psi\in\mathcal{H}$ \cite{BenZyc06,Fre12}. In this metric
there is only one geodesic between two pure states unless they are maximally
remote \cite[Theorem 1]{Ken84}. We can also identify $\mathbb{P}\mathcal{H}$
with the set $\mathcal{P}\left(  \mathcal{H}\right)  $ of one-dimensional
projections in $\mathcal{H}$ by sending $\left[  \varphi\right]  \rightarrow
P_{\varphi}:=\left|  \varphi\right\rangle \left\langle \varphi\right|
/\left\langle \varphi|\varphi\right\rangle $, where $\left|  \varphi
\right\rangle \left\langle \varphi\right|  $ denotes the orthogonal projection
operator onto the subspace generated by $\varphi\in\mathcal{H}$ (Dirac
notation). The transferred metric on $\mathcal{P}\left(  \mathcal{H}\right)
$, also called the \textsl{Fubini-Study metric}, is given by $D_{FS}\left(
\rho,\sigma\right)  :=\arccos\sqrt{\operatorname{tr}\left(  \rho\sigma\right)
}$ for $\rho,\sigma\in\mathcal{P}\left(  \mathcal{H}\right)  $. By
$\mathcal{S}\left(  \mathcal{H}\right)  $ we denote the convex closure of
$\mathcal{P}\left(  \mathcal{H}\right)  $, that is the set of density
(positive semi-definite, and trace one) operators on $\mathcal{H}$,
interpreted as $\textsl{mixed}$ $\textsl{states}$ of the system. Note that
$\dim_{\mathbb{R}}\mathcal{P}\left(  \mathcal{H}\right)  =2d-2$ and
$\dim_{\mathbb{R}}\mathcal{S}\left(  \mathcal{H}\right)  =d^{2}-1$. By
$m_{FS}$ we denote the unique unitarily invariant measure on $\mathbb{CP}^{d-1}$
or, equivalently, on $\mathcal{P}\left(  \mathbb{C}^{d}\right)  $.

The mixed states can be also described as elements of a ($d^{2}-1$)-dimensional
real Hilbert space (in fact, a Lie algebra) $\mathcal L_s^0(\mathcal H)$ of
Hermitian traceless operators on $\mathcal{H}$,
endowed with the Hilbert-Schmidt product given by $\left\langle \left\langle
\sigma,\tau\right\rangle \right\rangle _{HS}:=\operatorname{tr}\left(
\sigma\tau\right)  $ for $\sigma,\tau \in \mathcal L_s^0(\mathcal H)$.
Namely, the map defined by $b:\mathcal{S}\left(  \mathcal{H}\right)  \ni
\rho\rightarrow\rho-I/d \in \mathcal L_s^0(\mathcal H)$ gives us an affine
embedding (the \textsl{generalized Bloch representation}) of the set of mixed
(resp. pure) states into the ball (resp. sphere) in $\mathcal L_s^0(\mathcal H)$
of radius $\sqrt{1-d^{-1}}$, called the \textsl{generalized Bloch
ball} (resp. the \textsl{Bloch sphere}). Note, that the map $A\mapsto iA$ allows
us to identify $\mathcal L_s^0(\mathcal H)$ with  $\mathfrak{su}(d)$, the Lie
algebra of $\textrm{SU}(d)$, consisting of traceless skew-adjoint operators.
Only for $d=2$ the map is onto, and
for $d>2$ its image (the \textsl{Bloch vectors}) constitute a `thick' though
proper subset of ($d^{2}-1$)-dimensional ball, containing the (maximal) ball
of radius $1/\sqrt{d\left(  d-1\right)  }$ centered at $0$. On the other hand,
for $d>2$, $B\left(  d\right)  :=b\left(  \mathcal{P}\left(  \mathcal{H}%
\right)  \right)  $, the image of the space of pure states via $b$,
constitutes a `thin' ($2d-2$)-dimensional submanifold of the ($d^{2}%
-2$)-sphere. The metric spaces $\left(  \mathbb{P}\mathcal{H},D_{FS}\right)  $
and $\left(  B\left(  d\right)  ,D_{B}\right)  $, where $D_{B}$ is the great
arc distance on the Bloch sphere, though non-isometric for $d>2$, nevertheless
are ordinally equivalent, as the distances $D_{FS}$ and $D_{B}$ are related by
the formula $D_{B}\left(  b\left(  \left|  \varphi\right\rangle \left\langle
\varphi\right|  \right)  ,b\left(  \left|  \psi\right\rangle \left\langle
\psi\right|  \right)  \right)  =\gamma\left(  D_{FS}\left(  \left[
\varphi\right]  ,\left[  \psi\right]  \right)  \right)  $ ($\varphi,\psi
\in\mathcal{H}$), where a convex function $\gamma:\left[  0,\pi/2\right]
\rightarrow\mathbb{R}^{+}$ is given by $\gamma\left(  x\right)  =\sqrt
{1-d^{-1}}\arccos\frac{d\cos^{2}x-1}{d-1}$ for $0\leq x\leq\pi/2$. In other
words, scalar products of state vectors in $\mathbb{C}^{d}$ and their images
in $\mathbb{R}^{d^{2}-1}$ fulfill the relation: $\left|  \left\langle
\varphi|\psi\right\rangle \right|  ^{2}=\left\langle \left\langle b\left(
\left|  \varphi\right\rangle \left\langle \varphi\right|  \right)  ,b\left(
\left|  \psi\right\rangle \left\langle \psi\right|  \right)  \right\rangle
\right\rangle _{HS}+1/d$.

With a measurement of the system with a finite number $k$ of possible outcomes
one can associate a \textsl{positive operator valued measure} (\textsl{POVM})
defining the probabilities of the outcomes. A finite POVM is an ensemble of
positive semi-definite non-zero operators $\Pi_{j}$ ($j=1,\ldots,k$) on
$\mathcal{H}$ that sum to the identity operator, i.e. $\sum\nolimits_{j=1}%
^{k}\Pi_{j}=\mathbb{I}$. If the state of the system before the measurement
(the \textsl{input state}) is $\rho$, then the probability $p_{j}\left(
\rho\right)  $ of the $j$-th outcome is given by the Born rule, $p_{j}\left(
\rho\right)  =\operatorname{tr}\left(  \rho\Pi_{j}\right)  $. In general
situation, there is an infinite number of completely positive maps
(\textit{measurement instruments} in the sense of Davies and Lewis \cite{DavLew70})
describing conditional state changes due to the measurement and producing the
same measurement statistics, see \cite[Ch.\ 5]{HeiZem11}.
Among them, the \textit{efficient instruments} \cite{FucJac01} have particulary simple
form: they are given by the solutions of the set of
equations $\Pi_{j}=A_{j}^{\ast}A_{j}$ ($j=1,\ldots,k$), where $A_{j}$ are
bounded operators on $\mathcal{H}$. If $\rho$ is the input state and the
measurement outcome is $j$, then the state of the system after the measurement
is $\rho_{j}^{post}=A_{j}\rho A_{j}^{\ast}/p_{j}\left(  \rho\right)  $.
If, additionally, $A_{j}=\sqrt{\Pi_{j}}$ we get so called
\textit{generalised L\"{u}ders instrument} disturbing the initial
state in the minimal way \cite[p.~404]{DecGra07}.

A special class of POVMs are \textsl{normalized} \textsl{rank-}$\textsl{1}$
POVMs, where $\Pi_{j}$ ($j=1,\ldots,k$) are rank-$1$ operators and
$\operatorname{tr}\left(  \Pi_{j}\right)  =\operatorname{const}(j)=d/k$.
Necessarily, $k\geq d$ in this case, and there exists an ensemble of pure
states $\sigma_{j}\in\mathcal{P}\left(  \mathcal{H}\right)  $ ($j=1,\ldots,k$)
such that $\Pi_{j}=\left(  d/k\right)  \sigma_{j}$. Thus, $\sum\nolimits_{j=1}%
^{k}\sigma_{j}=\left(  k/d\right)  \mathbb{I}$, and so a normalized rank-$1$
POVM can be also defined as a (multi-)set of points in $\mathcal{P}\left(
\mathcal{H}\right)  $ that constitutes a \textsl{uniform }(or
\textsl{normalized}) \textsl{tight frame} in $\mathcal{P}\left(
\mathcal{H}\right)  $ \cite{EldFor02,BenFic03,Casetal13}, that is an ensemble
that fulfills $\sum\nolimits_{j=1}^{k}\operatorname{tr}\left(  \sigma_{j}%
\rho\right)  =k/d$ for every $\rho\in\mathcal{P}\left(  \mathcal{H}%
\right)  $. In this case we shall say that $\sigma_{j}$ ($j=1,\ldots,k$)
constitute a POVM. Equivalently, we can define normalized rank-$1$ POVMs as
complex projective $1$-designs, where by a \textsl{complex projective }$t$%
\textsl{-design} ($t\in\mathbb{N}$) we mean an ensemble $\left\{  \sigma
_{j}:j=1,\ldots,k\right\}  $ such that
\begin{equation}
\frac{1}{k^{2}}\sum_{j,m=1}^{k}f\left(  \operatorname{tr}\left(  \sigma_{j}%
\sigma_{m}\right)  \right)  =\int_{\mathcal{P}\left(  \mathbb{C}^{d}\right)
}\int_{\mathcal{P}\left(  \mathbb{C}^{d}\right)  }f\left(  \operatorname{tr}%
\left(  \rho\sigma\right)  \right)  dm_{FS}\left(  \rho\right)  dm_{FS}\left(
\sigma\right)
\end{equation}
for every $f:\mathbb{R\rightarrow R}$ polynomial of degree $t$ or less
\cite{Sco06}. The equality $\sum\nolimits_{j=1}^{k}\sigma_{j}=\left(
k/d\right)  \mathbb{I}$ is in turn equivalent to $\sum\nolimits_{j=1}%
^{k}b\left(  \sigma_{j}\right)  =0$, which gives the following simple
characterization of normalized rank-$1$ POVMs in the language of Bloch vectors:

\begin{proposition}
The generalized Bloch representation gives a one-to-one correspondence between
finite normalized rank-$1$ POVMs and finite (multi-)sets of points in
$B\left(  d\right)  $ with its center of mass at $0$.
\end{proposition}

The probabilities of the measurement outcomes in the generalized Bloch
representation take the form
\begin{equation}
\label{probabilities}
p_{j}\left(\rho\right)
=(d/k)\operatorname{tr}\left(  \sigma_{j}\rho\right)  =(d\cdot\left\langle
\left\langle b\left(  \sigma_{j}\right)  ,b\left(  \rho\right)  \right\rangle
\right\rangle _{HS}+1)/k
\end{equation}
for $\rho\in\mathcal{P}\left(  \mathbb{C}%
^{d}\right)  $ and $j=1,\ldots,k$. Obviously, the probability of obtaining
$j$-th outcome vary from 0, when the initial state is orthogonal to $\sigma_j$,
to $d/k$\label{bound}, when it coincides with $\sigma_j$. In consequence, any
outcome cannot be certain for given input state unless the measurement is
projective (in which case $k=d$).

\section{Symmetric, resolving and highly symmetric sets in metric
spaces}
\label{SRHS}

In this section we present a framework to investigate the concept of symmetry
in metric spaces. Let us start from general definition. Let $S$ be a subset of
a homogeneous metric space $\left(  X,r\right)  $, i.e.\ the group of all isometries (surjective maps
preserving metric $r$) acts transitively on $X$, that is for every $x,y\in X$ there exists an isometry
$f:X\to X$ such that $f(x)=y$. By
$\operatorname{Sym}\left(  S\right)$ we denote the group of symmetries of $S$,
that is, the group of all isometries leaving $S$ invariant. We call $S$ \textsl{symmetric}
if $\operatorname{Sym}\left(  S\right)$ acts transitively on $S$.

We say that $S$ is a \textsl{resolving set} \cite{DezDez13} if and only if $r\left( a,x\right)
= r\left(  b,x\right)$ for every $x\in S$ implies $a=b$, for $a,b\in X$.
The following proposition belongs to folklore:

\begin{proposition}
\label{folklore}If $S$ is resolving, then $f|_{S}=g|_{S}$
implies $f=g$ for every $f,g\in\operatorname{Sym}\left(  S\right)$.
Moreover, if $S$ is finite, then $\operatorname{Sym}\left(  S\right)$ is finite.
\end{proposition}

\begin{proof}
Let $f,g\in\operatorname{Sym}\left(  S\right)  $, $f|_{S}=g|_{S}$, and $a\in
X$. Then, for every $x\in S$ we have $r(fa,x)=r\left(  a,f^{-1}x\right)
=r\left(  a,g^{-1}x\right)  =r\left(  ga,x\right)  $. Hence $fa=ga$. Now, if
$\left|  S\right|  =k$, then $\operatorname{Sym}\left(  S\right)  $ is a
subgroup of the symmetric group $S_{k}$, and so is finite.
\end{proof}

To single out sets of higher symmetry we have to recall some notions from the
general theory of group action, see e.g. \cite{Fie07}. Let $G$ be a group
acting on $X$. For $x\in X$ we define its \textsl{orbit} as $Gx:=\left\{
gx:g\in G\right\}  $ and its \textsl{stabilizer} (or \textsl{isotropy
subgroup}) $G_{x}$ as the set of elements in $G$ that fix $x$, i.e.
$G_{x}:=\left\{  g\in G:gx=x\right\}  $. Obviously, two points lying on the
same orbit have conjugate stabilizers, since $G_{gx}=gG_{x}g^{-1}$ for $x\in
X$ and $g\in G$. The points of $X$ with the same stabilizers up to a conjugacy
are said to be of the same \textsl{isotropy type}, which is a measure of
symmetry of points (orbits). The points of the same isotropy type as $x$ form
the \textsl{orbit stratum }$\Sigma_{x}$. The decomposition of $X$ into orbit
strata is called the \textsl{orbit stratification}. Clearly, it induces a
stratification of the orbit space $X/G$. The natural partial order on the set
of all conjugacy classes of subgroups of $G$ induces the order on the set of
strata, namely, $\Sigma_{x}\prec\Sigma_{y}$ if and only if there exists $g\in
G$ such that $G_{x}\subset gG_{y}g^{-1}$ for $x,y\in X$, so that the maximal
strata consist of points with maximal stabilizers.

Assume now that a non-empty finite set $S\subset X$ is symmetric and consider
the action of the group $\operatorname{Sym}\left( S\right)$ on $X$. Clearly,
the whole set $S$ is contained in one orbit and hence in one stratum. We shall
say that $S$ is \textsl{highly symmetric} if and only if this stratum is maximal.
The following proposition gives a simple sufficient condition for the high symmetry.

\begin{proposition}
\label{prisym}
If $\operatorname*{Sym}(S)$ acts primitively on $S$ (i.e. the only
$\operatorname*{Sym}(S)$-invariant partitions of $S$ are trivial)
and the set of its common fixed points in $X$ is empty, then $S$ is highly symmetric.
\end{proposition}

\begin{proof}
Put $G:=\operatorname{Sym}\left( S\right)$. As primitive action of $G$ on $S$ must be transitive,
so $S$ is symmetric. Assume that $S$ is not highly symmetric. Then there
exist $x \in S$ and $y\in X \setminus S$ such that $G_{x}\subsetneq G_{y}$. It follows from
the primitivity of $G$ that $G_{x}$ is its maximal subgroup \cite[Corollary 8.14]{Isa08}.
Hence $G_{y}=G$, a contradiction.
\end{proof}

If $\operatorname{Sym}\left(  S\right)$ acts \textsl{doubly transitively} on $S$
\cite[p.~225]{Isa08}, i.e., if for every
$x_{1},x_{2},y_{1},y_{2}\in S$, $x_{1}\neq x_{2}$ and $y_{1}\neq y_{2}$ there
is $g\in\operatorname*{Sym}(S)$ such that $g\left(  x_{i}\right) = y_{i}$ for
$i=1,2$, we shall call such a set \textsl{super-symmetric}
after \cite{Zhu14,Zhu15b}. It is well-known that doubly transitive group action
is primitive \cite[Lemma 8.16]{Isa08}. Hence we get

\begin{corollary}
\label{supersym}
If $S\subset X$ is super-symmetric and the set of common fixed points of
$\operatorname*{Sym}(S)$ is empty, then $S$ is highly symmetric.
\end{corollary}

Let $S\subset X$ be symmetric. We say that $\kappa:S\rightarrow X$ is
$\operatorname{Sym}\left(  S\right)  $-\textsl{equivariant} if and only if
$g\kappa(x)=\kappa(gx)$ for all $g\in\operatorname{Sym}\left(  S\right)  $ and for some
(and hence for all) $x\in S$. For such $\kappa$ we call $\kappa\left(  S\right)  $ a
\textsl{factor of }$S$. Note that in this case $\operatorname{Sym}\left(  S\right)
\subset\operatorname{Sym}\left(  \kappa\left(  S\right)  \right)  $ and
$G_{x}\subset G_{\kappa\left(  x\right)  }$ for every $x\in S$.
A symmetric set is highly symmetric if and only if it does not have
a non-trivial factor:

\begin{proposition}
\label{higsymequ}Let $S\subset X$ be symmetric. Then $S$ is highly symmetric
if and only if every $\operatorname{Sym}\left(  S\right)  $-equivariant map
$\kappa:S\rightarrow X$ is one-to-one.
\end{proposition}

\begin{proof}
If $\left|  S\right|  =1$, then the proposition is trivial, as every singleton is
highly symmetric. Assume that $\left|  S\right|  \geq2$ and put
$G:=\operatorname{Sym}\left(  S\right)$.
If $S$ is not highly symmetric, then there exist $x\in S=Gx$ and $y\notin S$
such that $G_{x}\subsetneq G_{y}$. Put $\kappa(gx):=g(y)$ for every $g\in G$.
Clearly, $\kappa$ is well defined, $\operatorname*{Sym}\left(  S\right)
$-equivariant and it is not one-to-one, since otherwise $G_{y}\subset G_{x}$,
which is a contradiction. On the other hand, take a $G$-equivariant map
$\kappa:S\rightarrow X$ that is not one-to-one. Then, there exist $x\in S$ and
$g\in\operatorname{Sym}\left(  S\right)  $ such that $x\neq gx$ and $\kappa \left(
x\right)  =\kappa \left(  gx\right)  =g  \kappa \left(x\right)  $, and so $G_{x}\subsetneq
G_{\kappa \left(  x\right)  }$, a contradiction.
\end{proof}

It is interesting that an analogous idea was explored almost fifty years ago
by Zajtz, who defined so called \textsl{primitive geometric objects} in quite
similar manner as highly symmetric sets defined above and proved the fact
parallel to Proposition \ref{higsymequ} \cite[Theorem 1]{Zaj66}.

\section{Symmetric, informationally complete and highly symmetric normalized
rank-1 POVMs}
\label{HSPOVM}

To apply these general definitions to normalized rank-$1$ POVMs, note that
from the celebrated Wigner theorem \cite{Wig31} it follows that for every
separable Hilbert space $\mathcal{H}$ the group of isometries of homogeneous metric space $\left(
\mathcal{P}\left(  \mathcal{H}\right)  ,D_{FS}\right)  $ (\textsl{quantum
symmetries}) is isomorphic to the projective unitary-antiunitary group
$\operatorname{PUA}\left(  \mathcal{H}\right)  $, consisting of unitary and
antiunitary transformations of $\mathcal{H}$ defined up to phase factors, see
also \cite{Casetal97,Casetal03,KelPapRey08,GraKusMar09,Fre12}. To be more precise,
each such isometry is given by the map $\sigma_{U}:\mathcal{P}\left(
\mathcal{H}\right)  \ni\rho\rightarrow U\rho U^{\ast}\in\mathcal{P}\left(
\mathcal{H}\right)  $ for a unitary or antiunitary $U$, and two such
isometries coincide if and only if the corresponding transformations differ
only by a phase. Equivalence classes of unitary isometries form a normal subgroup of
$\operatorname{PUA}\left(  \mathcal{H}\right)  $ of index $2$, namely the
projective unitary group $\operatorname*{PU}\left(  \mathcal{H}\right)  $.
Clearly, every such isometry can be uniquely extended to a continuous affine
map on $\mathcal{S}\left(  \mathcal{H}\right)  $.

If $\mathcal{H}=\mathbb{C}^{d}$, then the generalized Bloch representation
gives a one-to-one correspondence between the compact group
$\operatorname{PUA}\left(  d\right)  $ and the group of isometries of the unit
sphere in $(d^{2}-1)$-dimensional real vector space $\mathcal L_s^0(\mathbb C^d)$
endowed with the Hilbert-Schmidt product, whose action
leaves the Bloch set $b(\mathcal S(\mathbb C^d))  $ invariant. This
correspondence is given by $[U] \rightarrow \left\{\rho\rightarrow U\rho
U^{\ast}:\rho\in\mathcal L_s^0(\mathbb C^d) \right\} $ for $U\in\operatorname{UA}\left(d\right)
$ (the unitary case is shown in \cite{ArrPat03} and it can be easily generalized
to the antiunitary case). Hence $\operatorname{PUA}\left(  d\right) $ is isomorphic
to a subgroup of the orthogonal group $O\left(  d^{2}-1\right)  $. Moreover,
$m_{FS}$ is the unique $\operatorname{PUA}\left(  d\right)  $-invariant
measure on $\mathcal{P}\left(  \mathbb{C}^{d}\right)  \simeq\mathbb{CP}^{d-1}%
$. In particular, for $d=2$, we have $\operatorname{PUA}\left(  2\right)
\simeq O\left(  3\right)  $, and so all quantum symmetries of qubit states can
be interpreted as rotations (for unitary symmetries, as $\operatorname*{PU}%
\left(  2\right)  \simeq SO\left(  3\right)  $), reflections or
rotoreflections of the three dimensional Euclidean space.

Taking this into account we can transfer the notions of symmetry and high
symmetry from $\mathcal{P}\left(  \mathbb{C}^{d}\right)  $ to finite
normalized rank-$1$ POVMs in $\mathbb{C}^{d}$. Let $\Pi=(\Pi_{j}%
)_{j=1,\ldots,k}$ be a finite normalized rank-$1$ POVM in $\mathbb{C}^{d}$ and
$S$ be a corresponding set of pure quantum states. We say that
\begin{itemize}
\item $\Pi$ is a \textsl{symmetric POVM} $\Leftrightarrow$ $S$ is symmetric in
$(\mathcal{P}\left(  \mathbb{C}^{d}\right)  ,D_{FS})$;

\item $\Pi$ is a \textsl{highly symmetric POVM} (\textsl{HS-POVM}%
)\textsl{\ }$\Leftrightarrow$ $S$ is highly symmetric in $(\mathcal{P}\left(
\mathbb{C}^{d}\right)  ,D_{FS})$.
\end{itemize}
For finite normalized rank-$1$ measurements symmetric POVMs coincide with
\textsl{group covariant POVMs} introduced by Holevo \cite{Hol82} and studied
since then by many authors. We say that a measurement $\Pi=(\Pi_{j}%
)_{j=1,\ldots,k}$ is $G$\textsl{-covariant} for a group $G$ if and only if
there exists $G\ni g\rightarrow\sigma_{U_{g}}\in\operatorname{PUA}\left(
d\right)  $, a \textsl{projective unitary-antiunitary representation} of $G$
(i.e. a homomorphism from $G$ to $\operatorname{PUA}\left(  d\right)  $), and
a surjection $s:G\rightarrow\left\{  1,\ldots,k\right\}  $ such that
$\sigma_{U_{g_1}}(\Pi_{s( g_2)})=U_{g_1}\Pi_{s(g_2)}
U_{g_1}^{\ast}=\Pi_{s(g_1g_2)}$ for all $g_1,g_2\in G$.
For the greater convenience, we can assume that $\Pi$ is a multiset,
and so we can label its elements by $g$ instead of $s(g)$: $\Pi=(\Pi_g)_{g\in G}$.
In order to guarantee that $\sum_{g\in G}\Pi_g=\mathbb I$ we need to put
$\Pi_g=(|s(G)|/|G|)\Pi_{s(g)}$. Let $\Pi$ be a
finite normalized rank-$1$ POVM in $\mathbb{C}^{d}$ and $S$ be a corresponding
set of pure quantum states. It is clear that a symmetric finite normalized
rank-$1$ POVM is $\operatorname{Sym}\left(  S\right)  $-covariant, and,
conversely, if a finite normalized rank-$1$ POVM  is $G$-covariant, then
$(  \sigma_{U_{g}})  _{g\in G}$ is a subgroup of the group of
isometries of  $(\mathcal{P}(\mathbb{C}^{d}),D_{FS})$, acting transitively on the corresponding (multi-)set of
pure states. We call the representation \textsl{irreducible} if and only if
$\mathbb{I}/d$ is the only element of $\mathcal{S}(\mathbb{C}^{d})$ invariant under action of the representation. It follows from
the version of Schur's lemma for unitary-antiunitary maps \cite[Theorem~II]{Dim63}
that this definition coincides with the classical one.
Irreducibility of the representation can be also equivalently expressed as
follows: for any pure state $\tau\in\mathcal{P}(\mathbb{C}^{d})$ its orbit under the action of the representation generates a rank-$1$
$G$-covariant POVM, i.e.\ $\frac{1}{|G|}\sum_{g\in G}\sigma_{U_{g}}\left(
\tau\right)  =\mathbb{I}/d$, see also \cite{ValWal05}.

In the next section we shall describe all HS-POVMs in dimension 2. From Corollary \ref{supersym} and \cite[Theorem 1]{Zhu14},
we already know that the SIC-POVM in dimension two (represented by
a tetrahedron), the Hesse SIC-POVM in dimension three, and the set of $64$ Hoggar
lines in dimension eight are highly symmetric POVMs, see also \cite{Zhu15}.
Note that our definition of highly symmetric POVMs resembles the definition of highly
symmetric frames introduced by Broome and Waldron \cite{Bro10,BroWal13,Wal13}.
However, they consider subsets of $\mathbb{C}^{d}$ rather than $\mathbb{CP}%
^{d-1}$ and unitary symmetries rather than projective unitary-antiunitary symmetries.

The next proposition clarifies the relations between the properties of the set
of pure states constituting a finite normalized rank-$1$ POVM and the
properties of its Bloch representation. We call a normalized rank-$1$ POVM
$\Pi=(\Pi_{j})_{j=1,\ldots,k}$ \textsl{informationally complete} (resp.
\textsl{purely informationally complete}) if and only if the probabilities
$p_{j}\left(  \rho\right)  $ ($j=1,\ldots,k$) determine uniquely every input
state $\rho\in\mathcal{S}\left(  \mathbb{C}^{d}\right)$ (resp.
$\mathcal{P}\left(  \mathbb{C}^{d}\right)$).  Since we need $d^2-1$ independent
parameters to describe uniquely a quantum state, any IC-POVM must contain at
least $d^2$ elements\label{ic}. The following result provides necessary and
sufficient conditions for informational completeness and purely informational
completeness:

\begin{proposition}
\label{infcom}Let $\Pi=(\Pi_{j})_{j=1,\ldots,k}$ be a finite normalized
rank-$1$ POVM in $\mathbb{C}^{d}$ and $S:=\left\{  \sigma_{j}:j=1,\ldots
,k\right\}  $ be a corresponding set of pure quantum states, i.e.\ $\sigma
_{j}\in\mathcal{P}(  \mathbb{C}^{d})  $ and $\Pi_{j}=\left(
d/k\right)  \sigma_{j}$ for $j=1,\ldots,k$. Let us consider the following properties:

\begin{enumerate}
\item[(a)] $S$ is a complex projective $2$-design;
\item[(b)] $b\left(  S\right)  $ is a normalized tight frame in $\mathcal L_s^0(\mathbb C^d) $;
\item[(c)] $b\left(  S\right)  $ is a spherical $2$-design in $\mathcal L_s^0(\mathbb C^d)  $;
\item[(d)] $\Pi$ is informationally complete;
\item[(e)] $b\left(  S\right)  $ generates $\mathcal L_s^0(\mathbb C^d) $;
\item[(f)] $b\left(  S\right)  $ is a frame in $\mathcal L_s^0(\mathbb C^d)$;
\item[(g)] $\Pi$ is purely informationally complete;
\item[(h)] $S$ is a resolving set in $(\mathcal{P}(\mathbb{C}^{d})  ,D_{FS})$;
\item[(i)] $b\left(  S\right)  $ is a resolving set in $\left(  B\left(
d\right)  ,D_{B}\right)  $.
\end{enumerate}

Then $(a)\Leftrightarrow(b)\Leftrightarrow(c)\Rightarrow(d)\Leftrightarrow(e)
\Leftrightarrow(f)\Rightarrow(g)\Leftrightarrow(h)\Leftrightarrow(i)$.
Moreover, if $d=2$, then $(g)\Rightarrow(d)$.
\end{proposition}

\begin{proof}
It is obvious that $(b)\Rightarrow(f)$ and $(d)\Rightarrow(g)$. The proof of
$(a)\Leftrightarrow(b)$ can be found in \cite[Proposition 13]{Sco06},
$(b)\Leftrightarrow (c)$ in \cite[p.~5]{Wal03} and
$(d)\Leftrightarrow(e)$ in \cite[Proposition 3.51]{HeiZem11}. It is well known
that in finite dimensional spaces frames are generating sets, hence
$(e)\Leftrightarrow(f)$. Furthermore, $(g)\Leftrightarrow(h)\Leftrightarrow(i)$
follows from the fact that the distances $D_{FS}$ and $D_{B}$ are ordinally
equivalent, and from the equality $\operatorname{tr}\left(  \rho\sigma\right)
=\cos^{2}D_{FS}\left(  \rho,\sigma\right)  $ for $\rho,\sigma\in
\mathcal{P}\left(  \mathbb{C}^{d}\right)  $. Moreover, for $d=2$ the notions of purely
informational completeness and informational completeness coincide
\cite[Remark~1]{Heietal13}.
\end{proof}

A POVM that satisfies $(a)$ (or, equivalently, $(b)$ or $(c)$) is called
\textsl{tight informationally complete POVM} \cite{Sco06}.
Note that $(d)$ does not imply $(b)$, even if $S$ is symmetric and $d=2$. To
show this, consider $S\subset\mathcal{P}\left(  \mathbb{C}^{2}\right)  $ such
that $b(S)=\{  2^{-1/2}(e_{1} \pm e_{2}),$ $ 2^{-1/2}(- e_{1} \pm e_{3}
)\}  $, where $\left\{  e_{1},e_{2},e_{3}\right\}  $ is any orthonormal
basis of $\mathcal L_s^0(\mathbb C^2)$. Then $b\left(  S\right)  $ is a
tetragonal disphenoid with the antiprismatic
symmetry group $D_{2d}$. Clearly, $b\left(  S\right)  $ is a frame in
$\mathcal L_s^0(\mathbb C^2) $, but simple calculations show that it is not
tight. On the other hand, one can prove $(d)\Rightarrow(b)$, under the additional
assumption that the natural action of $\operatorname*{Sym}\left(  S\right)$
on $\mathcal L_s^0(\mathbb C^d) $ is irreducible, applying
\cite[Theorem 6.3]{ValWal05}. Moreover, as we shall see in the next section,
all the conditions above are equivalent if $S$ is highly symmetric and $d=2$.

\section{Classification of highly symmetric POVMs in dimension
two}
\label{HSPOVM2}

\begin{theorem}
There are only eight types of HS-POVMs in two dimensions, seven exceptional
informationally complete HS-POVM represented in $\mathbb{R}^{3}$ by five
Platonic solids (convex regular polyhedra): the \textbf{tetrahedron}, \textbf{cube},
\textbf{octahedron}, \textbf{icosahedron} and \textbf{dodecahedron} and two convex
quasi-regular polyhedra: the \textbf{cuboctahedron} and \textbf{icosidodecahedron},
and an infinite series of non informationally complete HS-POVMs represented in
$\mathbb{R}^{3}$ by \textbf{regular polygons}, including \textbf{digon}.
\end{theorem}

\begin{proof}
Let $S=\left\{  \sigma_{j}:j=1,\ldots,k\right\}  \subset\mathcal{P}\left(
\mathbb{C}^{2}\right)  \simeq\mathbb{CP}$ constitute a HS-POVM, and let
$B:=b(S)\subset S^{2}$. Put $G:=\operatorname{Sym}\left(  B\right)  $. Then it
follows from the equivalence $(d)\Leftrightarrow (e)$ in Proposition \ref{infcom}
that either $B$ is contained in a proper (one- or two-dimensional) subspace of
$\mathbb{R}^{3}$, or the POVM is informationally complete and, according to the
implication $(d)\Rightarrow (h)$ in Proposition \ref{infcom} and
Proposition \ref{folklore}, $G$~is finite.

If $G$ is infinite, then necessarily the stabilizer of any element $x\in B$
has to be infinite, since otherwise the whole orbit of $x$ would be infinite.
As the only linear isometries of $\mathbb{R}^{3}$ leaving possibly $x$
invariant are either rotations about the axis $l_{x}$ through $x $, or
reflections in any plane containing $l_{x}$, the stabilizer $G_{x}$ has to
contain an infinite subgroup of rotations about $l_{x}$. Thus the orbit of any
point beyond $l_{x}$ under $G$ must be infinite. In consequence, $B=\left\{
-x,x\right\}  $, and $G=D_{\infty h}\simeq O\left(  2\right)  \times C_{2}$.

If $G$ is finite, it must be one of the point groups, i.e., finite subgroups
of $O\left(  3\right)  $. The complete characterization of such subgroups has
been known for very long time \cite{Sen90}: there exist seven infinite
families of axial (or prismatic) groups $C_{n}$, $C_{nv}$, $C_{nh}$, $S_{2n}$,
$D_{n}$, $D_{nd}$ and $D_{nh}$, as well as seven additional polyhedral (or
spherical) groups: $T$ (chiral tetrahedral), $T_{d}$ (full tetrahedral),
$T_{h}$ (pyritohedral), $O$ (chiral octahedral), $O_{h}$ (full octahedral),
$I$ (chiral icosahedral) and $I_{h}$ (full icosahedral). Analysing their
standard action on $S^{2}$ (see e.g.
\cite{RobCar70,LimMonRob01,MicZhi01,Zhi01,Quietal05}), one can find in all
cases the orbits with maximal stabilizers. Gathering this information
together, we get all highly symmetric finite subsets of $S^{2}$, and so all
HS-POVMs in two dimensions. These sets are listed in Tab.~\ref{HSdim2} together with
their symmetry groups and the stabilizers of their elements with respect to
these symmetry groups. For all but the first two types of HS-POVMs, the
symmetry group $G$ is a polyhedral group, and so it acts irreducibly on
$\mathbb{R}^{3}$. Hence, $B$ must be a tight frame in all these cases.
\qedhere
\begin{table}
\caption{HS-POVMs in dimension two, with their cardinalities, symmetry groups and stabilizers of elements (in Schoenflies notation).}
\label{HSdim2}
\centering
\begin{tabular}{lccc}
\hline\noalign{\smallskip}
convex hull of the orbit & cardinality of the orbit & group &
stabilizer  \\
\noalign{\smallskip}\hline\noalign{\smallskip}
digon & $2$ & $D_{\infty h}$ & $C_{\infty v}$ \\
regular $n$-gon ($n\geq3$) & $n$ & $D_{nh}$ & $C_{2v}$ \\
tetrahedron & $4$ & $T_{d}$ & $C_{3v}$ \\
octahedron & $6$ & $O_{h}$ & $C_{4v}$ \\
cube & $8$ & $O_{h}$ & $C_{3v}$ \\
cuboctahedron & $12$ & $O_{h}$ & $C_{2v}$ \\
icosahedron & $12$ & $I_{h}$ & $C_{5v}$ \\
dodecahedron & $20$ & $I_{h}$ & $C_{3v}$ \\
icosidodecahedron & $30$ & $I_{h}$ & $C_{2v}$\\
\noalign{\smallskip}\hline
\end{tabular}
\end{table}

\end{proof}

We have just shown that if $S\subset\mathcal P(\mathbb C^2)$ constitutes
an informationally complete HS-POVM in dimension two, then $b(S)$ is
a spherical 2-design. However, it follows from \cite[Theorem~2]{CraPerKri10}
and the form of corresponding group invariant polynomials (listed in
Sect.~\ref{invpol}) that if $\textrm{Sym}(b(S))=O_h$, then $b(S)$ is a spherical
3-design and if $\textrm{Sym}(b(S))=I_h$, then $b(S)$ is a spherical 5-design.

Classification of all finite symmetric subsets of $S^{2}$ and, in consequence,
all symmetric normalized rank-1 POVMs in two dimensions, is of course more
complicated than for highly symmetric case. In particular, the number of such
non-isometric subsets is uncountable. However, since each symmetric subset
generates a \textsl{vertex-transitive polyhedron} in three-dimensional
Euclidean space (and each such polyhedron is a symmetric set generating
symmetric normalized rank-1 POVM), the task reduces to classifying
such polyhedra, which was done by Robertson and Carter in the 1970s, see
\cite{RobCar70,Robetal70,Rob84,Cro97}. They proved that the transitive
polyhedra in $\mathbb{R}^{3}$ can be parameterized (up to isometry) by metric
space (with the Hausdorff distance under the action of Euclidean isometries
related closely to the Gromov-Hausdorff distance, see \cite{Mem08}), which is
a two-dimensional CW-complex with $0$-cells corresponding exactly to highly
symmetric subsets of $S^{2}$.

Note that not only `regular polygonal' POVMs (e.g.\ the trine or
`Mercedes-Benz' measurement for $k=3$ \cite{Jozetal03} and the `Chrysler'
measurement for $k=5$ \cite{Wil13}), but also `Platonic solid' POVMs have
been considered earlier by several authors in various quantum mechanical
contexts, including quantum tomography, at least since 1989 \cite{Jon89,Jon91b}, see for instance \cite{Cab03,Cavetal04,DecJanBet04,Bur08}.

\section{Entropy and relative entropy of measurement}
\label{ENTRELENT}

\subsection{Definition}
\label{Def}

Let $\Pi=(\Pi_{j})_{j=1,\ldots,k}$ be a finite POVM in $\mathbb{C}^{d}$. We
shall look for the most `classical' (with respect to a given measurement)
or `coherent' quantum states, i.e. for the states that minimize the uncertainty
of the outcomes of the measurement. This uncertainty can be measured by the
quantity called the \textsl{entropy of measurement} given by
\begin{equation}
H(\rho,\Pi):=\sum_{j=1}^{k}\eta\left(  p_{j}\left(  \rho,\Pi\right)  \right)
\text{,}%
\end{equation}
for $\rho\in\mathcal{S}\left(  \mathbb{C}^{d}\right)  $, where the probability
$p_{j}\left(  \rho,\Pi\right)  $ of the $j$-th outcome ($j=1,\ldots,k$) is
given by $p_{j}\left(  \rho,\Pi\right)  :=\operatorname{tr}\left(  \rho\Pi
_{j}\right)  $, and the \textsl{Shannon entropy function} $\eta:\left[
0,1\right]  \rightarrow\mathbb{R}^{+}$ by $\eta\left(  x\right)  :=-x\ln x$
for $x>0$, and $\eta\left(  0\right)  :=0$. (In the sequel, we shall use
frequently the identity $\eta\left(  xy\right)  =\eta\left(  x\right)
y+\eta\left(  y\right)  x$, $x,y\in\left[  0,1\right]  $.) Thus, the entropy
of measurement $H(\rho,\Pi)$ is just the Boltzmann-Shannon entropy of the
probability distribution of the measurement outcomes, assuming that the state
of the system before the measurement was $\rho$. This quantity (as well as its
continuous analogue) has been considered by many authors, first in the 1960s
under the name of \textsl{Ingarden-Urbanik entropy} or $A$\textsl{-entropy},
then, since the 1980s, in the context of \textsl{entropic uncertainty
principles} \cite{Deu83,KriPar02,Mas07,WehWin10}, and also quite recently for
more general statistical theories \cite{Slo03,ShoWer10}. Wilde called it the
\textsl{Shannon entropy of POVM} \cite{Wil13}. For a history of this notion,
see \cite{Weh78} and \cite{Baletal86}.

The function $H(\cdot,\Pi):\mathcal{S}\left(\mathbb{C}^{d}\right)  \rightarrow\mathbb{R}$
is continuous and concave. In consequence, it attains minima in the set of pure states.
Moreover, it is obviously bounded from above by $\ln k$, the entropy of the uniform distribution,
and the upper bound is achieved for the maximally mixed state $\rho_{\ast}:=\mathbb I/d$.
The general bound from below is expressed with the help of the \textsl{von Neumann entropy} of
the state $\rho$ given by $S(\rho):=-\operatorname{tr}(\rho\ln\rho)$ \cite[Sect.~2.3]{Lanetal11}:
\begin{equation}
S\left(  \rho\right)-\sum_{j=1}^k p_j\ln(\operatorname{tr} (\Pi_j)) \leq H(\rho,\Pi)\leq\ln k
\text{.}
\end{equation}
Since for the normalized rank-$1$ POVM $\operatorname{tr}(\Pi_j)=d/k$ for all $j=1,\ldots,k$,
we get
\begin{equation}
S\left(  \rho\right)+\ln(k/d)\leq H(\rho,\Pi)\leq\ln k
\text{.}
\label{entine}
\end{equation}
(Moreover for $\rho\in\mathcal{S}\left(  \mathbb{C}^{d}\right)  $,
$S\left(\rho\right)  =\min H(\rho,\Pi)$, where the minimum is taken over all
normalized rank-$1$ POVMs $\Pi$, see, e.g.\ \cite[Sect.~11.1.2]{Wil13}.)
In consequence, for $\rho \in \mathcal P(C^d)$ we have
\begin{equation}
\ln (k/d) \leq H(\rho,\Pi) \leq \ln k.
\label{entine2}
\end{equation}
The first inequality in (\ref{entine2}) follows also from the inequalities
$p_j(\rho,\Pi) \leq d/k$ for every $j=1,\ldots,k$, and from the fact that
$\ln$ is an increasing function.

It is sometimes much more convenient to work with the \textsl{relative entropy
of measurement }(\textsl{with respect to the uniform distribution}) \cite[p.~67]{Gre11}
that measures non-uniformity of the distribution of the measurement outcomes and is
given by
\begin{equation}
\widetilde{H}(\rho,\Pi):=\ln k-H(\rho,\Pi)\text{,}
\end{equation}
and to look for the states that maximize this quantity. Clearly, it follows from
(\ref{entine}) that the relative entropy of measurement is bounded from below
by $0$, and from above by the relative von Neumann entropy of the state $\rho$
with respect to the maximally mixed state $\rho_{\ast}=I/d$:%
\begin{equation}
0\leq\widetilde{H}(\rho,\Pi)\leq S\left(  \rho|\rho_{\ast}\right)  \leq\ln
d\text{.}%
\end{equation}

\subsection{Relation to informational power}
\label{Infpow}

The problem of minimizing entropy (and so maximizing relative entropy) is
connected with the problem of maximization of the mutual information between
ensembles of initial states (classical-quantum states) and the POVM $\Pi$.

Let us consider an ensemble $\mathcal{E}=\left(  \left(  \tau_{i}\right)
_{i=1}^{m},\left(  p_{i}\right)  _{i=1}^{m}\right)  $, where $p_{i}\geq0$ are
\emph{a priori} probabilities of density matrices $\tau_{i}\in\mathcal{S}%
\left(  \mathbb{C}^{d}\right)  $, where $i=1,\ldots,m$, and $\sum
\nolimits_{i=1}^{m}p_{i}=1$. The \textsl{mutual information} between
$\mathcal{E}$ and $\Pi$ is given by:%

\begin{equation}
I\left(  \mathcal{E},\Pi\right)  :=\sum\limits_{i=1}^{m}\eta\left(
\sum\limits_{j=1}^{k}P_{ij}\right)  +\sum\limits_{j=1}^{k}\eta\left(
\sum\limits_{i=1}^{m}P_{ij}\right)  -\sum\limits_{j=1}^{k}\sum\limits_{i=1}%
^{m}\eta\left(  P_{ij}\right)
\end{equation}
where $P_{ij}:=p_{i}\operatorname{tr}\left(  \tau_{i}\Pi_{j}\right)  $ is the
probability that the initial state of the system is $\tau_{i}$ and the
measurement result is $j$ for $i=1,\ldots,m$ and $j=1,\ldots,k$.

The problem of maximization of $I(\mathcal{E},\Pi)$ consists of two dual
aspects \cite{Arnetal14a,Hol12,Hol13}: maximization over all possible
measurements, providing the ensemble $\mathcal{E}$ is given, see,
e.g.\ \cite{Hol73,Dav78,Sasetal99,Suzetal07}, and (less explored) maximization
over all ensembles, when the POVM $\Pi$ is fixed \cite{Arnetal11,Oreetal11}.
In the former case, the maximum is called \textsl{accessible information}. In
the latter case, Dall'Arno et al.\ \cite{Arnetal11,Arnetal14a} introduced the
name \textsl{informational power of }$\Pi$ for the maximum and denoted it by
$W\left(  \Pi\right)  $. Dall'Arno et al.\ \cite{Arnetal11} and,
independently, Oreshkov et al.\ \cite{Oreetal11} showed that there always exists a
\textsl{maximally informative ensemble} (i.e.\ ensemble that maximizes the
mutual information) consisting of pure states only. Note that
a POVM $\Pi$ generates a \textsl{quantum-classical channel} $\Phi
:\mathcal{S}\left(  \mathbb{C}^{d}\right)  \rightarrow\mathcal{S}\left(
\mathbb{C}^{k}\right)  $ given by $\Phi\left(  \rho\right)  =\sum_{j=1}%
^{k}\operatorname{tr}\left(  \rho\Pi_{j}\right)  \left|  e_{j}\right\rangle
\left\langle e_{j}\right|  $, where $\left(  \left|  e_{j}\right\rangle
\right)  _{j=1}^{k}$ is any orthonormal basis in $\mathbb{C}^{k}$. The
\textsl{minimum output entropy} of $\Phi$ is equal to the minimum entropy of
$\Pi$, i.e.\ $\min_{\rho\in\mathcal{P}\left(  \mathbb{C}^{d}\right)  }%
S(\Phi(\rho))=\min_{\rho\in\mathcal{P}\left(  \mathbb{C}^{d}\right)  }%
H(\rho,\Pi)$ \cite{Sho02}. On the other hand, the informational power of $\Pi$
can be identified \cite{Arnetal11,Oreetal11,Hol12b} as the \textsl{classical (Holevo)
capacity} $\chi(\Phi)$ of the channel $\Phi$, i.e.\
\[
W(\Pi)=\chi(\Phi):=\max_{\left(  \left(  \tau_{i}\right)  _{i=1}^{m},\left(
p_{i}\right)  _{i=1}^{m}\right)  }\left\{  S\left(  \sum\limits_{i=1}^{m}%
p_{i}\Phi(\tau_{i})\right)  -\sum\limits_{i=1}^{m}p_{i}S(\Phi(\tau
_{i}))\right\}  \text{.}%
\]

What are the relation between informational power and entropy minimization? It
follows from \cite[p. 2]{Hol12b} that%
\[
I\left(  \mathcal{E},\Pi\right)  =H\left(  \sum\limits_{i=1}^{m}p_{i}\tau
_{i},\Pi\right)  -\sum\limits_{i=1}^{m}p_{i}H\left(  \tau_{i},\Pi\right)
\text{.}%
\]
Clearly, $H\left(  \sum\nolimits_{i=1}^{m}p_{i}\tau_{i},\Pi\right)  \leq\ln k$
and for $i=1,\ldots,k$ we have $\min\limits_{\rho\in\mathcal{P}\left(
\mathbb{C}^{d}\right)  }H\left(  \rho,\Pi\right)  \leq H\left(  \tau_{i}%
,\Pi\right)  $. Hence%
\begin{equation}
W\left(  \Pi\right)  \leq\ln k-\min\limits_{\rho\in\mathcal{P}\left(
\mathbb{C}^{d}\right)  }H\left(  \rho,\Pi\right)  =
\max\limits_{\rho\in\mathcal{P}\left(
\mathbb{C}^{d}\right)  }\widetilde{H}\left(  \rho,\Pi\right)\text{.}
\label{ineinfpow}%
\end{equation}

In consequence, the equality in (\ref{ineinfpow}) holds if and only if there
exists an ensemble $\mathcal{E}=\left(  \left(  \tau_{i}\right)  _{i=1}%
^{m},\left(  p_{i}\right)  _{i=1}^{m}\right)  $ such that $\operatorname{tr}%
\left(  \left(  \sum\nolimits_{i=1}^{m}p_{i}\tau_{i}\right)  \Pi_{j}\right)
=1/k$ for $j=1,\ldots,k$ and $\tau_{1},\ldots,\tau_{m}\in\arg\min H$.

Assume now that $\Pi$ is normalized $1$-rank POVM with $\Pi_{j}=(d/k)\sigma
_{j}$, for $j=1,\ldots,k$, where $S:=\left\{  \sigma_{j}:j=1,\ldots,k\right\}
\subset\mathcal{P}(  \mathbb{C}^{d})  $. Then, applying the
Carath\'{e}odory convexity theorem, we can characterize the situation, where
the two maximization problems coincide:

\begin{proposition}
The following two conditions are equivalent:

\begin{enumerate}
\item $\left\langle b\left(  S\right)  \right\rangle ^{\perp}\cap
\operatorname*{conv}\left(  b\left(  \arg\min H\right)  \right)  \neq
\emptyset$;
\item  the equality in (\ref{ineinfpow}) holds.
\label{coninfpow}
\end{enumerate}

Moreover, if $\Pi$ is informationally complete, then (1) can be replaced by
\linebreak $0\in\operatorname*{conv}\left(  b\left(  \arg\min H\right)  \right)  $.
\end{proposition}

In particular, condition (1) is fulfilled if $\Pi$ (and so $S$) is symmetric
or if $\operatorname*{Sym}(S)$ acts irreducibly on $\mathcal P(\mathbb C^d)$. (The latter is true, if e.g.
$d=2$ and $S$ is a union of pairs of orthogonal states.) To see this, it is
enough to consider the ensemble consisting of equiprobable elements of the
orbit of any minimizer of $H$ under the action of $\operatorname*{Sym}(S)$ and
use the fact that $(1/\left|  \operatorname*{Sym}(S)\right|  )\sum
\nolimits_{g\in\operatorname*{Sym}(S)}\left(  g\sigma\right)  =I/d$ for all
$\sigma\in S$. This fact was observed already by Holevo in \cite{Hol05}.

\subsection{Relation to entropic uncertainty principles}
\label{Eup}

The entropic uncertainty principles form another area of research related to quantifying the uncertainty in quantum theory.
They were introduced  by Bia\l ynicki-Birula and Mycielski \cite{BiaMyc75}, who showed that they are stronger than `standard'
Heisenberg's uncertainty principle, and Deutsch \cite{Deu83}, who provided the first lower bound for the sum of entropic uncertainties
of two observables independent on the initial state. This bound has been later improved by Maassen and Uffink \cite{MaaUff88,Uff94},
and de Vicente and S\'anchez-Ruiz \cite{VicSan07}. The generalizations for POVMs (all previous results referred to PVMs) has been
formulated subsequently in \cite{Hal97,KriPar02,Mas07} and \cite{Ras08}. More detailed survey of the topic can be found in \cite{WehWin10}.

The entropic uncertainty relations are closely connected with entropy minimization.
In fact, any lower bound for the entropy of
measurement can be regarded as an \textsl{entropic uncertainty relation for
single measurement} \cite{KriPar02}. Moreover, combining $m$ normalized
rank-$1$ $k$-element POVMs $\Pi^{i}=(\Pi_{j}^{i})_{j=1,\ldots,k}$ ($i=1,\ldots,m$) we
obtain another normalized rank-$1$ $km$-element POVM $\Pi:=(\frac{1}{m}\Pi_{j}%
^{i})_{j=1,\ldots,k}^{i=1,\ldots,m}$. Now, from an entropic uncertainty
principle for $\left(  \Pi^{i}\right)  _{i=1,\ldots,m}$ written in the form
$\frac{1}{m}\sum\nolimits_{i=1}^m H(\rho,\Pi^{i})\geq C>0$ \cite[p.~3]{WehWin10}
we get automatically a lower bound for entropy of $\Pi$, namely
\begin{equation}
H(\rho,\Pi)=\frac{1}{m}\sum_{i=1}^m H(\rho,\Pi^{i})+\ln m\geq
C+\ln m \label{eupsingle}
\end{equation}
for $\rho\in\mathcal{P}\left(  \mathbb{C}^{d}\right)$, and vice versa, proving a lower bound for entropy of $\Pi$ we get immediately an entropic uncertainty principle.

To be more specific, assume now that $m=2$ and $\Pi_{j}^{i}=\left(  d/k\right)  \sigma
_{j}^{i}$, where $\sigma_{j}^{i}=\left|  \varphi_{j}^{i}\right\rangle
\left\langle \varphi_{j}^{i}\right|  \in\mathcal{P}\left(  \mathbb{C}%
^{d}\right)  $, denoting their Bloch vectors by $x_{j}^{i}:=b\left(  \sigma
_{j}^{i}\right)  \in\mathcal L_s^0(\mathbb C^d)  \simeq\mathbb{R}^{d^{2}-1}$
for  $j=1,\ldots,k$, $i=1,2$. The \emph{Krishna-Parthasarathy entropic uncertainty
principle} \cite[Corollary~2.6]{KriPar02}, combined with (\ref{eupsingle})  gives us
\begin{align}
H(\rho,\Pi)&\geq\ln\left(  2k/d\right)  -\ln\max_{j,l=1,\ldots,k}\left|  \left\langle
\varphi_{j}^{1}|\varphi_{l}^{2}\right\rangle \right|  \\
& =\ln\left(  2k/d\right)  -\frac{1}{2}\ln (  \max_{j,l=1,\ldots,k}
\left\langle \left\langle x_{j}^{1},x_{l}^{2}\right\rangle \right\rangle
_{HS} +1/d )   \nonumber
\end{align}
for $\rho\in\mathcal{P}\left(  \mathbb{C}^{d}\right)  $. In consequence, taking into account
that the radius of the Bloch sphere is $\sqrt{1-1/d}$, we get  an upper bound for relative entropy
\begin{align}
\widetilde{H}(\rho,\Pi)  &  \leq\ln d+\frac{1}{2}\ln(\max_{j,l=1,\ldots
,k}  \left\langle \left\langle x_{j}^{1},x_{l}^{2}\right\rangle
\right\rangle _{HS}+1/d) \label{uppbou}\\
&  =\ln d+\frac{1}{2}\ln((  1-1/d)  (\max_{j,l=1,\ldots
,k}  \cos\left(  2\theta_{jl}\right)    )+1/d))\text{,}\nonumber
\end{align}
where $\theta_{jl}:=\measuredangle\left(  x_{j}^{1},x_{l}^{2}\right)  /2$ for
$j,l=1,\ldots,k$. As this upper bound does not depend on the input state
$\rho$, it gives us also an upper bound for the informational power of $\Pi$.

If $d=2$, this inequality takes a simple form%
\begin{equation}
\widetilde{H}(\rho,\Pi)\leq\ln2+\ln\max_{j,l=1,\ldots,k}\left|\cos\theta_{jl}\right|   \text{.}
\label{uppboudim2}%
\end{equation}
We may used this bound, e.g.\ for the `rectangle' POVM analysed in
Sect.~\ref{Local}, that can be treated as the aggregation of two pairs of
antipodal points on the sphere representing two PVM measurements. In this case
we deduce from (\ref{uppboudim2}) that $\widetilde{H}\leq\ln2+\ln\max\left(  \left|
\sin\left(  \alpha/2\right)  \right|  ,\left|  \cos\left(  \alpha/2\right)
\right|  \right)  $, where $\alpha$ is the measure of the angle between the
diagonals of the rectangle. In particular, for the `square' POVM we get
$\widetilde{H}\leq\frac{1}{2}\ln2$. As we shall see in Sect.~\ref{INFPOW}, this
bound is actually reached for each of four states constituting the POVM and
represented by the vertices of the square.

\subsection{Relation to Wehrl entropy minimization}
\label{Weh}

Let us consider now a normalized rank-1 POVM $\Pi=(\Pi_{j})_{j=1,\ldots,k}$  in
$\mathbb{C}^{d}$ with $\Pi_{j}=\left(
d/k\right)  \sigma_{j}$ for $j=1,\ldots,k$, where $S:=\left\{  \sigma_{j}:j=1,\ldots,k\right\} \subset\mathcal P(\mathbb C^d) $.
Then we get after simple calculations%
\begin{equation}
H(\rho,\Pi)=\frac{d}{k}\sum_{j=1}^{k}\eta\left(  \operatorname{Tr}\left(
\rho\sigma_{j}\right)  \right)  -\ln\left(  d/k\right)  \label{ent}%
\end{equation}
and so
\begin{equation}
\widetilde{H}(\rho,\Pi)=\ln d-\frac{d}{k}\sum_{j=1}^{k}\eta\left(
\operatorname{Tr}\left(  \rho\sigma_{j}\right)  \right)  \label{relent}%
\end{equation}
for $\rho\in\mathcal{P}(  \mathbb{C}^{d})  $. Assume now, that
$\Pi$ is symmetric (group covariant) and put $G:=\operatorname{Sym}\left(
S\right)  $. Then for each $\tau\in S$ we have $S=\left\{  g\tau:g\in
G\right\}  $ and
\begin{align}
\widetilde{H}(\rho,\Pi)  &  =
\ln d-\frac{d}{\left|  S\right|  }\sum_{\left[  g\right]  \in G/G_{\tau}%
}\eta\left(  \operatorname{Tr}\left(  \rho\left(  g\tau\right)  \right)
\right)
\label{relentbis} \\
&  =
\ln d-\frac{d}{\left|  G\right|  }\sum_{g\in
G}\eta\left(  \operatorname{Tr}\left(  \rho\left(  g\tau\right)  \right)
\right)
\nonumber
\end{align}
for $\rho\in\mathcal{P}(  \mathbb{C}^{d})  $. The same formulae are
true for any subgroup of $\operatorname{Sym}\left(  S\right)  $ acting
transitively on $S$. Note that the behaviour of the functions $H(\cdot
,\Pi),\ \widetilde{H}(\cdot,\Pi):\mathcal{P}(  \mathbb{C}^{d})
\rightarrow\mathbb{R}^{+}$ depends only on the choice of the fiducial state $\tau$.
Moreover, observe that both functions are $G$-invariant, as for $\rho\in\mathcal{P}(\mathbb{C}^{d})$
and $g \in G$ we have $\operatorname{Tr}(\rho (g \tau)) = \operatorname{Tr} (\tau (g^{-1} \rho))$,
and so from (\ref{relentbis}) we get
\begin{equation}
H(\rho,\Pi)=\frac{d}{\left|  G\right|  }\sum_{g\in G}\eta\left(  \operatorname{Tr}\left(  \tau\left(  g\rho\right)  \right)
\right) -\ln(d/k)
\end{equation}
and
\begin{equation}
\widetilde{H}(\rho,\Pi)=\ln d-\frac{d}{\left|  G\right|  }\sum_{g\in G}%
\eta\left(  \operatorname{Tr}\left(  \tau\left(  g\rho\right)  \right)
\right)  \text{.}%
\end{equation}

Now, one can observe that relative entropy of symmetric POVM is
closely related to the semi-classical quantum entropy introduced in 1978 by
Wehrl for the harmonic oscillator coherent states \cite{Weh78} and named later
after him. The definition was generalised by Schroeck \cite{Sch85}, who
analysed its basic properties. Let $G$ be a compact topological group acting
unitarily and irreducibly on $\mathcal{P}\left(  \mathbb{C}^{d}\right)  $.
Fixing fiducial state $\tau\in\mathcal{P}\left(  \mathbb{C}^{d}\right)  $ we
get the family of states $\left(  g\tau\right)  _{g\in G/G_{\tau}}$ called
(\textsl{generalized} or \textsl{group}) \textsl{coherent states}
\cite{Per86,Alietal14} that fulfills the identity: $\int_{G/G_{\tau}}g\tau
d\mu([g]_{G/G_{\tau}})=\mathbb{I}$, where $\mu$~is the
$G$-invariant measure on $G/G_{\tau}$ such that $\mu\left(  G/G_{\tau}\right)
=d$. Then for $\rho\in\mathcal{S}\left(  \mathbb{C}^{d}\right)  $ we define
the \textsl{generalized Wehrl entropy} of $\rho$ by
\begin{equation}
S_{Wehrl}\left(  \rho\right)  :=\int_{G/G_{\tau}}\eta\left(
\operatorname*{tr}\left(  \rho\left(  g\tau\right)  \right)  \right)
d\mu(\left[  g\right]  _{G/G_{\tau}})\text{.}%
\end{equation}
It is just the Boltzmann-Gibbs entropy for the density function on $\left(
G/G_{\tau},\mu\right)  $ called the \textsl{Husimi function} of $\rho$ and
given by $G/G_{\tau}\ni\left[  g\right]  _{G/G_{\tau}}\rightarrow
\operatorname*{tr}\left(  \rho\left(  g\tau\right)  \right)  \in\mathbb{R}%
^{+}$ that represents the probability density of the results of an approximate
coherent states measurement (or in other words continuous POVM)
\cite{Dav76,BusSch89}. Then the relative Boltzmann-Gibbs entropy of the
Husimi distribution of $\rho$ with respect to the Husimi distribution of the
maximally mixed state $\rho_{\ast}$, that is the\ constant density on $\left(
G/G_{\tau},\mu\right)  $ equal $1/d$, given by%
\begin{equation}
S_{Wehrl}\left(  \rho|\rho_{\ast}\right)  :=\ln d-S_{Wehrl}\left(
\rho\right)
\end{equation}
is a continuous analogue of $\widetilde{H}(\cdot,\Pi)$ given by (\ref{relentbis}).
What is more, the
relative entropy of measurement is just a special case of such transformed
Wehrl entropy, when we consider the discrete coherent states (i.e.\ POVM)
generated by a finite group. On the other hand, the entropy of measurement
$H(\cdot,\Pi)$ has no continuous analogue, as it may diverge to infinity, where
$k\rightarrow\infty$. In principle, to define coherent states we can use an
arbitrary fiducial state. However, to obtain coherent states with sensible
properties one has to choose the fiducial state $\tau$ to be the
\textsl{vacuum state}, that is the state with maximal symmetry with respect to
$G$ \cite{Kly07},\cite[Sect.~2.4]{Per86}.

To investigate the Wehrl entropy it is enough to require that $G$ should be
locally compact. In fact, Wehrl defined this quantity for the \textsl{harmonic
oscillator coherent states}, where $G$ is the Heisenberg-Weyl group $H_{4}$
acting on projective (infinite dimensional and separable) Hilbert space, $G_{\tau}\simeq
U\left(  1\right)  \times U\left(  1\right)  $, and $G/G_{\tau}\simeq
\mathbb{C}$. This notion was generalized by Lieb \cite{Lie78} to \textsl{spin
}(\textsl{Bloch})\textsl{\ coherent states}, with $G=SU(2)$ acting on
$\mathbb{CP}^{d-1}$ ($d\geq2$), $G_{\tau}\simeq U(1)$ and $G/G_{\tau}\simeq
S^{2}$. In this paper Lieb proved that for harmonic oscillator coherent states
the minimum value of the Wehrl entropy is attained for coherent states
themselves. (It follows from the group invariance that this quantity is the
same for each coherent state.) He also conjectured that the statement is true
for spin coherent states, but, despite many partial results, the problem,
called the \textsl{Lieb conjecture}, had remained open for next thirty five
years until it was finally proved by Lieb himself and by Solovej in 2012
\cite{LieSol12}. They also expressed the hope that the same result holds for
$SU(N)$ coherent states for arbitrary $N\in\mathbb{N}$, or even for any
compact connected semisimple Lie group (the \textsl{generalized Lieb
conjecture}), see also \cite{GnuZyc01,Slo03}. Bandyopadhyay received recently
some partial results in this direction for $G=SU(1,1)$ coherent states
\cite{Ban09}, where $G_{\tau}\simeq U(1)$ and $G/G_{\tau}$ is the hyperbolic plane.

For finite groups and covariant POVMs the minimization of Wehrl entropy is
equivalent to the maximization of the relative entropy of measurement, which
is in turn equivalent to the minimization of the entropy of measurement.
Consequently, one could expect that the entropy of measurement should be
minimal for the states constituting the POVM that are already known to be
critical as inert states.
We shall see in Sec.~\ref{Gloext} that this need not be always the case.
In particular, it is not true for the tetrahedral POVM or in the situation where
the states constituting  a POVM form a regular polygon with odd number of vertices.
Thus, it is conceivable that to prove the `generalized Lieb conjecture'
some additional assumptions will be necessary.

\subsection{Relation to quantum dynamical entropy}
\label{QDE}

As in the preceding section, let $\Pi=(\Pi_{j})_{j=1,\ldots,k}$ be a finite
normalized rank-$1$ POVM in $\mathbb{C}^{d}$ and let $S=\left\{  \sigma
_{j}:j=1,\ldots,k\right\}  $ be a corresponding (multi-)set of pure quantum
states. Set $\sigma_{j}=\left|  \varphi_{j}\right\rangle \left\langle
\varphi_{j}\right|  $, where $\varphi_{j}\in\mathbb{C}^{d}$, $\left\|
\varphi_{j}\right\|  =1$. Assume that successive measurements described by the
generalized L\"uders instrument connected with $\Pi$, where $\sigma_i$ serve as the
`output states', are performed on an evolving quantum system and that the motion
of the system between two subsequent measurements is governed by a unitary matrix
$U$. Clearly, the sequence of measurements introduces a nonunitary evolution
and the complete dynamics of the system can be described by a quantum
Markovian stochastic process, see \cite{SloZyc94}.

The results of consecutive measurements are represented by finite strings of
letters from a $k$-element alphabet. Probability of obtaining the string
$\left(  i_{1},\ldots,i_{n}\right)  $, where $i_{j}=1,\ldots,k$ for
$j=1,\ldots,n$ and $n\in\mathbb{N}$ is then given by%
\begin{equation}
P_{i_{1},\ldots,i_{n}}\left(  \rho\right)  :=p_{i_{1}}\left(  \rho\right)
\cdot%
{\textstyle\prod\nolimits_{m=1}^{n-1}}
p_{i_{m}i_{m+1}}\text{,} \label{seqpro}%
\end{equation}
where $\rho$ is the initial state of the system, $p_{i}\left(
\rho\right) :=\left(  d/k\right)\operatorname{tr}\left(\rho\sigma_{i}\right)$
is the probability of obtaining $i$ in the first measurement,
and $p_{ij}:=\left(  d/k\right)  \operatorname{tr}\left(  U\sigma_{i}U^{\ast
}\sigma_{j}\right)  =\left(  d/k\right)  \left|  \left\langle \varphi
_{j}|U|\varphi_{i}\right\rangle \right|  ^{2}$ is the probability of getting $j$
as the result of the measurement, providing the result of the preceding
measurement was $i$, for $i,j=1,\ldots,k$ \cite{SloZyc94,Slo03}. The
randomness of the measurement outcomes can be analysed with the help of
(\emph{quantum}) \emph{dynamical entropy}, the quantity introduced for the
L\"{u}ders-von Neumann\ measurement independently by Srinivas \cite{Sri78},
Pechukas \cite{Pec82}, Beck \& Graudenz \cite{BecGra92} and many others, see
\cite[p.~5685]{SloZyc94}, then generalized by \.{Z}yczkowski and one of the
present authors (W.S.) to arbitrary classical or quantum measurements and
instruments \cite{SloZyc94,Kwaetal97,SloZyc98,Slo03}, and recently
rediscovered by Crutchfield and Wiesner under the name of
\textsl{quantum entropy rate} \cite{CruWie08}.

The definition of (\emph{quantum}) \emph{dynamical entropy of }$U$\textsl{ with
respect to }$\Pi$ mimics its classical counterpart, the Kolmogorov-Sinai
entropy:%
\begin{equation}
H\left(  U,\Pi\right)  :=\lim_{n\rightarrow\infty}(H_{n+1}-H_{n})=\lim
_{n\rightarrow\infty}H_{n}/n\text{,} \label{dynentdef}%
\end{equation}
where $H_{n}:=\sum_{i_{1},...,i_{n}=1}^{k}\eta\left(  P_{i_{1},\ldots,i_{n}%
}\left(  \rho_{\ast}\right)  \right)  $ for $n\in\mathbb{N}$. The maximally
mixed state $\rho_{\ast}=\mathbb I/d$ plays here the role of the `stationary state'
for combined evolution. It is easy to show that the quantity is given by%
\begin{align}
H\left(  U,\Pi\right)   &  =\frac{1}{k}\sum_{i,j=1}^{k}\eta\left(  \left(
d/k\right)  \operatorname{tr}\left(  U\sigma_{i}U^{\ast}\sigma_{j}\right)  \right)
\label{dynentfor}\\
&  =\ln\left(  k/d\right)  +\frac{d}{k^{2}}\sum_{i,j=1}^{k}\eta\left(
\operatorname{tr}\left(  U\sigma_{i}U^{\ast}\sigma_{j}\right)  \right) \nonumber\\
&  =\ln\left(  k/d\right)  +\frac{d}{k^{2}}\sum_{i,j=1}^{k}\eta\left(  \left|
\left\langle \varphi_{i}|U|\varphi_{j}\right\rangle \right|  ^{2}\right)
\text{,} \nonumber%
\end{align}
which is a special case of much more general integral entropy formula
\cite{Slo03}. Using (\ref{ent}) and (\ref{dynentfor}) we see that the
dynamical entropy of $U$ is expressed as the mean entropy of measurement over
output states transformed by $U$:%
\begin{equation}
H\left(  U,\Pi\right)  =\frac{1}{k}\sum_{i=1}^{k}H(U\sigma_{i}U^{\ast},\Pi)
\label{dynentrel}%
\end{equation}
There are two sources of randomness in this model: the underlying unitary
dynamics and the measurement process. The latter can be measured by the
quantity $H_{meas}\left(  \Pi\right)  :=H\left( \mathbb I,\Pi\right)  $ called
\textsl{(quantum) measurement entropy}. From (\ref{dynentrel}) we get%
\begin{equation}
H_{meas}\left(  \Pi\right)  =\frac{1}{k}\sum_{i=1}^{k}H(\sigma_{i},\Pi)\text{.}
\label{meaentrel}%
\end{equation}

If $\Pi$ is symmetric, then all the summands in (\ref{meaentrel}) are the
same. Hence, in this case, the measurement entropy $H_{meas}\left(
\Pi\right)  $ is equal to the entropy of measurement $H(\rho,\Pi)$, where the
input state $\rho$ is one of the output states from $S$.

\subsection{Entropy in the Bloch representation}
\label{EntropyBloch}

Using the Bloch representation for states and normalized rank-$1$ POVMs
(see Sect.~\ref{QSP}) one can reformulate the problems of entropy minimization and
relative entropy maximization  as the problems of finding the global extrema of the
corresponding function on $B(d)\subset S^{d^2-2}$. Such reformulation significantly reduces
the complexity of the problem, especially in dimension two, since in this case $B(2)$ is
isomorphic with $S^2$. Let $\Pi=(\Pi_{j})_{j=1,\ldots,k}$ be a normalized rank-$1$ POVM
in $\mathbb{C}^{d}$ such that
$\Pi_{j}=\left(  d/k\right)  \sigma_{j}$, $\sigma_{j}\in\mathcal{P}\left(
\mathbb{C}^{d}\right)$, and let $B := \left\{  v_{j}|j=1,\ldots,k\right\}$,
where $v_{j} := \sqrt{d/(d-1)}b(\sigma_{j})  \in S^{d^2-2}$ ($j=1,\ldots,k$).
For $\rho\in\mathcal{S}\left(  \mathbb{C}^{d}\right)  $, $u := \sqrt{d/(d-1)}b(\rho) \in B^{d^2-2}$
and $j=1,\ldots,k$  we get from (\ref{probabilities})
\begin{equation}
p_j(\rho,\Pi)=((d-1)u\cdot v_j+1)/k. \label{prob}
\end{equation}
 Applying (\ref{prob}), (\ref{ent}) and (\ref{relent}) we obtain
\begin{equation}
H_{B}(u):=H(\rho,\Pi)=\sum_{j=1}^{k}\eta\left(  \frac{(d-1)u\cdot v_{j}+1}%
{k}\right)  =\ln\frac{k}{d}+\frac{d}{k}\sum_{j=1}^{k}h\left(  u\cdot
v_{j}\right)
\label{entdimtwo}
\end{equation}
and
\begin{equation}
\widetilde{H}_{B}(u):=\widetilde{H}(\rho,\Pi)=\ln d-\frac{d}{k}\sum_{j=1}%
^{k}h\left(  u\cdot v_{j}\right)  \text{,}
\label{entreldimtwo}
\end{equation}
where the function  $h:\left[  -1/(d-1),1\right]  \rightarrow\mathbb{R}^{+}$ is given by%
\begin{equation}
h\left(  t\right)  :=\eta\left(  \frac{(d-1)t+1}{d}\right)  \label{h}%
\end{equation}
for $-1\leq t\leq1$. It is clear that the functions $H_{B}$ and $\widetilde
{H}_{B}$ restricted to \linebreak $\sqrt{d/(d-1)}B(d)$ are of $C^{2}$ class (even analytic)
except at the points `orthogonal' to the points from $B$,
in the sense that they represent orthogonal states to
states $\sigma_j$, $j=1,\ldots,k$. For $d=2$ they are just the points
antipodal to the points from $B$. Despite the fact that
the function $h$ is non-differentiable at $-1/(d-1)$, some standard calculations
show that in dimension two $H_B$ and $\widetilde H_B$ are at these points of $C^{1}$ class
but not twice differentiable.

For $d=2$, it follows from (\ref{entdimtwo}) that if $B$ is contained in a plane $L$,
then $H_{B}$ attains global minima on this plane. Indeed, for $u \in S^2$ we have
$H_B(u)=H_B(\widetilde{u})$, where $\widetilde{u}$ is the orthogonal projection of $u$
onto $L$, as $u\cdot v_{j}=\widetilde{u}\cdot v_{j}$ for $j=1,\ldots,k$.

For a symmetric POVM, there exists a finite group $G\subset O\left(
d^2-1\right)$ acting transitively on $B$. It follows from (\ref{entdimtwo})
and (\ref{entreldimtwo}) that $H_{B}$ and $\widetilde{H}_{B}$ are $G$-invariant
functions on $B(d)$ given by%
\begin{equation}
H_{v}(u):=H_{B}(u)=\ln\frac{\left|  Gv\right|  }{d}+\frac{d}{\left|  G\right|
}\sum_{g\in G}h\left(  gu\cdot v\right)  \label{entfor}
\end{equation}
and
\begin{equation}
\widetilde{H}_{v}(u):=\widetilde{H}_{B}(u)=\ln d-\frac{d}{\left|  G\right|
}\sum_{g\in G}h\left(  gu\cdot v\right)  \text{,} \label{relentfor}
\end{equation}
for $u\in B(d)$, where $v \in B=\left\{  gv:g\in G\right\}$ is
the normalized Bloch vector of an arbitrary fiducial state. This fact allows
us to use the theory of solving symmetric variational problems developed by
Louis Michel and others in the 1970s, and applied since then in many physical
contexts, especially in analysing the spontaneous symmetry breaking phenomenon
\cite{Mic80}.

\section{Local extrema of entropy for symmetric POVMs in dimension two}
\label{Local}

We start by quoting several theorems concerning smooth action of finite groups
on finite-dimensional manifolds. They are usually formulated for compact Lie
groups, but since finite groups are zero-dimensional Lie group, thus the
results apply equally well in this case.

Let $G$ be a finite group of $C^{1}$ maps acting on a compact
finite-dimensional manifold $M$. In the set of strata consider the order
$\prec$ introduced in Sect.~\ref{SRHS}.\ Then

\begin{Th*}
[Montgomery and Yang {\cite[Theorem~4a]{MicZhi01}, \cite{MonYan57,Fie07}}] The set of strata is
finite. There exists a unique minimal stratum, comprising elements of trivial
stabilizers, that is open and dense in $M$, called \textbf{generic} or
\textbf{principal}. For every $x\in M$ the set $\bigcup\left\{  \Sigma
_{u}:u\in M,\text{ }\Sigma_{x}\preceq\Sigma_{u}\right\}  $ is closed in $M$;
in particular, the maximal strata are closed.
\end{Th*}

The next result tells us, where we should look for the critical points of an
invariant function, i.e.\ the points where its gradient vanishes: we have to
focus on the maximal strata of $G$ action on $M$.

\begin{Th*}
[Michel {\cite[Corollary~4.3]{MicZhi01}, \cite{Mic71}}] Let $F:M\rightarrow\mathbb{R}$ be a
$G$-invariant $C^{1}$ map, and let $\Sigma$ be a maximal stratum. Then

\begin{enumerate}
\item $\Sigma$ contains some critical points of $F$;

\item  if $\Sigma$ is finite, then all its elements are critical points of $F
$.
\end{enumerate}
\end{Th*}

Such points are called \textsl{inert states}\label{inertstates} in physical literature; they are
critical regardless of the exact form of $F$, see, e.g.\ \cite{Xuetal12}. Of
course, an invariant function can have other critical points than those
guaranteed by the above theorem (\textsl{non-inert states}). However, we shall
see that for highly symmetric POVMs in dimension two, the global minima of
entropy function $H_{B}$ lie always on maximal strata. Although Michel's
theorem indicates a special character of the points with maximal stabilizer,
it does not give us any information about the nature of these critical points.
In some cases we can apply the following result:

\begin{Th*}
[Modern Purkiss Principle {\cite[p.~385]{Wat83}}] Let $F:M\rightarrow
\mathbb{R}$ be a $G$-invariant $C^{2}$ map and let $u\in M$. Assume that the
action of the linear isotropy group $\left\{  T_{u}h:h\in G_{u}\right\}  $ on
$T_{u}M$ is irreducible. Then $u$ is a critical point of $F$, which is either
degenerate (i.e.\ the Hessian of $F$ is singular at $u$) or a local extremum of
$F$.
\label{MPP}
\end{Th*}

For a finite group $G\subset \textrm{O}(2)$ acting irreducibly on the sphere $S^{2}$
points lying on its rotation axes form the maximal strata, and so, it follows from Michel's
theorem that they are critical for entropy functions. We can divide them into
three categories depending on whether they are antipodal to the elements of
the fiducial vector's orbit (type I), and, if not, whether their stabilizers
act irreducibly on the tangent space (type II) or not (type III). In the first
case, as well as, generically, in the second case, we can determine the character
of critical point using the following proposition:

\begin{proposition}
\label{criterion} Let $u\in S^{2}$ be a point lying on a rotation axis of the group $G\subset \textrm{O}(2)$
acting irreducibly on the sphere $S^{2}$ and let $v$ denote the normalized Bloch vector of an arbitrary fiducial
state for a rank-1 $G$-covariant POVM. Then:

\begin{enumerate}[1)]

\item If there exists $g\in G$ such that $u=-gv$, then $u$ is a local
minimizer (resp. maximizer) for $H_{v}$ (resp. $\widetilde{H}_{v}$);

\item If $u\neq-gv$ for every $g\in G$ and the linear isotropy group
$\left\{  T_{u}g:g\in G_{u}\right\}  $ acts irreducibly on $T_{u}S^{2}$ (or,
equivalently, $G_{u}$ contains a cyclic subgroup of order greater than $2$), then:

\begin{enumerate}
\item  if
\begin{equation}
\frac{2}{|G/G_{u}|}\sum_{\left[  h\right]  \in G/G_{u}}(hu\cdot v)\ln
(1+hu\cdot v)>1\text{,}%
\end{equation}
then $u$ is a local minimizer (resp. maximizer) for $H_{v}$ (resp.
$\widetilde{H}_{v}$),

\item  if
\begin{equation}
\frac{2}{|G/G_{u}|}\sum_{\left[  h\right]  \in G/G_{u}}(hu\cdot v)\ln
(1+hu\cdot v)<1\text{,}%
\end{equation}
then $u$ is a local maximizer (resp. minimizer) for $H_{v}$ (resp.
$\widetilde{H}_{v}$).
\end{enumerate}
\end{enumerate}
\end{proposition}

\begin{proof}
Fix any geodesic (i.e.\ a great circle)
passing by $u$. Let $q$ be one of two vectors lying on the intersection of the
plane orthogonal to $u$ passing through $0$ and the geodesic. As $0$ is the
only $G$-invariant vector in $\mathbb{R}^{3}$ and, at the same time, the only
$G_{u}$-invariant element orthogonal to $u$, we have $(1/|G_u|)\sum_{g\in G}gu=\sum_{\left[  h\right]
\in G/G_{u}}hu=0={\textstyle\sum_{g\in G_{u}}}gq$.
Consider a natural parametrisation of the great circle $\gamma
:(-\pi,\pi)\rightarrow\mathbb{S}^{2}$ (throwing away $-u$) given by
$\gamma\left(  \delta\right)  :=(\sin\delta)q+(\cos\delta)u$ for $\delta
\in(-\pi,\pi)$, where $\delta$ is the measure of the angle between vectors $u$
and $\gamma\left(  \delta\right)  $. Put $w:=\gamma\left(  \delta\right)  $.
Then it follows from (\ref{h}), (\ref{entfor}) and the equality $\sum_{g\in G} gw=0$ that%
\begin{align*}
\left(  H_{v}\circ\gamma\right)  \left(  \delta\right)   &  =H_{v}(w)\\
& = \ln\frac{|Gv|}{2}+\frac{2}{|G|} \sum_{g \in G} \eta( (1+gw\cdot v)/2)\\
&  =\ln{|Gv|}+\frac{1}{|G|}\sum_{g\in G}\eta(1+gw\cdot v)\\
&  =\ln{|Gv|}+\frac{1}{|G|}\sum_{\left[  h\right]  \in G/G_{u}}\sum_{g\in
Gu}\eta(1+hgw\cdot v)\\
&  =\ln{|Gv|}+\frac{1}{|G|}\sum_{\left[  h\right]  \in G/G_{u}}\sum_{g\in
Gu}\eta(1+(\sin\delta)hgq\cdot v+(\cos\delta)hu\cdot v)\\
&  =\ln{|Gv|}+\frac{1}{|G/G_{u}|}\sum_{\left[  h\right]  \in G/G_{u}}%
f_{h}(\delta)\text{,}%
\end{align*}
where
\[
f_{h}(\delta):=\frac{1}{|G_{u}|}\sum_{g\in G_{u}}\eta(1+(\sin\delta)hgq\cdot
v+(\cos\delta)hu\cdot v)\text{.}%
\]
Let $h\in G$ be such that $hu\neq-v$. Then, for $\delta$ small enough, we get
\begin{align*}
f_{h}^{\prime}(\delta)  &  =\frac{1}{|G_{u}|}\sum_{g\in G_{u}}\eta^{\prime
}(1+(\sin\delta)hgq\cdot v+(\cos\delta)hu\cdot v)\times\\
&  \times((\cos\delta)hgq\cdot v-(\sin\delta)hu\cdot v)\text{.}%
\end{align*}
In particular $f_{h}^{\prime}(0)=0$. Moreover,%
\begin{align*}
f_{h}^{\prime\prime}(\delta)  &  =\frac{1}{|G_{u}|}\sum_{g\in G_{u}}%
\eta^{\prime\prime}(1+(\sin\delta)hgq\cdot v+(\cos\delta)hu\cdot v)\times\\
&  \times((\cos\delta)hgq\cdot v-(\sin\delta)hu\cdot v)^{2}+\\
&  +\eta^{\prime}(1+(\sin\delta)hgq\cdot v+(\cos\delta)hu\cdot v)(-(\sin
\delta)hgq\cdot v-(\cos\delta)hu\cdot v))\text{.}%
\end{align*}
\smallskip

\emph{1}) Let $\widetilde{h}u=-v$ for some $\widetilde{h}\in G$. Then $\widetilde{h}gq\cdot v=0$ and so $f_{\widetilde{h}}(\delta)$ reduces to
$$f_{\widetilde{h}}(\delta)=1/|G_u|\sum_{g\in G_u}\eta(1-\cos\delta).$$
In consequence, for $\delta\neq 0$
\[
f_{\widetilde{h}}^{\prime}(\delta)=-(\ln(1-\cos\delta)+1)\sin\delta\text{,}%
\]
and so
\[
f_{\widetilde{h}}^{\prime\prime}(\delta)=-1-(\cos\delta)(\ln(1-\cos
\delta)+2)\text{.}%
\]
Since $f_{\widetilde{h}}^{\prime}(\delta)\to 0$ as $\delta\to 0$, so $f_{\widetilde{h}}^{\prime}(0)=0$.
Moreover, $f_{\widetilde{h}}^{\prime\prime}(\delta)\to \infty$ as $\delta\to 0$.
Let us observe that there exists $1>c>0$, such that the inequality $hu\cdot v\geq-1+c$
holds for any $[h]\in G/G_{u}$,$\ \left[  h\right]  \neq\lbrack\widetilde{h}]$.
Now we can estimate $|f_{h}^{\prime\prime}(\delta)|$ as follows:%
\begin{align*}
|f_{h}^{\prime\prime}(\delta)|  &  \leq\frac{1}{|G_{u}|}\sum_{g\in G_{u}%
}\left|  \frac{((\cos\delta)hgq\cdot v-(\sin\delta)hu\cdot v)^{2}}{1+hgw\cdot
v}\right| \\
&  +\frac{1}{|G_{u}|}\sum_{g\in G_{u}}\left|  (\ln(1+hgw\cdot v)+1)(hgw\cdot
v)\right| \\
&  \leq\frac{1}{|G_{u}|}\sum_{g\in G_{u}}\left(  \frac{1}{\left|  1+hgw\cdot
v\right|  }+(|\ln(1+hgw\cdot v)|+1)\right) \\
&  \leq f\left(  1-|\sin\delta|+(c-1)\cos\delta\right)  \text{,}%
\end{align*}
for $\left|  \delta\right|  <c$, where $f\left(  x\right)  :=\frac
{1}{\left|  x\right|  }+\left|  \ln x\right|  +1$ for $x>0$.
The last inequality follows from the fact that  $f$ is decreasing in
$(0,1)$, $1 + hgw\cdot v \geq  1 - |\sin \delta| + (c-1) \cos \delta$
and $1 - |\sin \delta| + (c-1) \cos \delta \geq 0$ for $|\delta| < c$.

Thus%
\begin{align*}
(H_{v}\circ\gamma)^{\prime\prime}(\delta)  &  =\frac{1}{|G/G_{u}%
|}\Big(f_{\widetilde{h}}^{\prime\prime}(\delta)+\sum_{[h]\in G/G_{u},[h]\neq
[\widetilde{h}]}f_{h}^{\prime\prime}(\delta)\Big)\\
&  \geq g\left(  \delta\right)  \overset{\delta\rightarrow0}{\longrightarrow
}+\infty,
\end{align*}
where%
\[
g\left(  \delta\right)  :=-\frac{1+(\cos\delta)(\ln(1-\cos\delta
)+2)+(|G/G_{u}|-1)f(1-\sin\delta+(c-1)\cos\delta)}{|G/G_{u}|}%
\]
for $\delta>0$. In particular, $(H_v \circ\gamma)'(0) = 0$ and there is $\varepsilon>0$ such that $(H_{v}%
\circ\gamma)^{\prime\prime}(\delta)>0$ for $\left|  \delta\right|
<\varepsilon$. Hence, one can find a neighbourhood $\mathcal{V}\subset S^{2}$ of
$u$ such that for any geodesic passing by $u$, $H_{v}$ is strictly convex on its part
contained in $\mathcal{V}$ and has minimum at $u$. Consequently,
$H_{v}(u)>H_{v}(w)$ for every $w\in\mathcal{V}$, $w\neq u$, which completes
the proof of 1).

\smallskip
\emph{2}) Assume that $u\neq-gv$ for every $g\in G$ and the linear isotropy group
$\left\{  T_{u}g:g\in G_{u}\right\}  $ acts irreducibly on $T_{u}S^{2}$. Then for every $h\in G$ we have
\begin{align*}
f_{h}^{\prime\prime}(0)  &  =\frac{1}{|G_{u}|}\sum_{g\in G_{u}}(\eta
^{\prime\prime}(1+hu\cdot v)(hgq\cdot v)^{2}+\eta^{\prime}(1+hu\cdot
v)(-hu\cdot v))\\
&  =\eta^{\prime}(1+hu\cdot v)(-hu\cdot v)+\eta^{\prime\prime}(1+hu\cdot
v)\frac{1}{|G_{u}|}\sum_{g\in G_{u}}(hgq\cdot v)^{2}\\
&  =(hu\cdot v)(\ln(1+hu\cdot v)+1)-\frac{1}{1+hu\cdot v}\frac{1}%
{2}(1-(hu\cdot v)^{2})\\
&  =(hu\cdot v)(\ln(1+hu\cdot v)+3/2)-1/2\text{,}%
\end{align*}
where the last but one identity follows from the fact that $\{hgq:g\in
G_{u}\}$ is  a~normalized tight frame in $S^{2}$ contained in the plane
orthogonal to $hu$ for each $h\in G_{u}$. Thus we obtain
\[
(H_{v}\circ\gamma)^{\prime\prime}(0)=\frac{1}{|G/G_{u}|}\sum_{\left[
h\right]  \in G/G_{u}}(hu\cdot v)\ln(1+hu\cdot v)-1/2
\]
and 2) follows from the Modern Purkiss Principle.
\end{proof}

If $\Pi$ is a HS-POVM, we can assume that $G$ is one of the following groups:
$D_{nh}$, $T_{d}$, $O_{h}$ or $I_{h}$, and the Bloch vector $v$ of the fiducial
vector lies in the maximal strata, consisting of points where the rotation
axes of the group intersect the Bloch sphere.  For $D_{nh}$ group
we have one $n$-fold and $n$ $2$-fold rotation axes ($2n+2$ points: a digon
and two regular $n$-gons); for $T_{d}$ group: three $2$-fold, four $3$-fold
rotation axes ($14$ points: an octahedron and two dual tetrahedra); for
$O_{h}$ group: six $2$-fold, four $3$-fold, three $4$-fold rotation axes ($26$
points: a~cuboctahedron, a~cube and an octahedron); for $I_{h}$ group: fifteen
$2$-fold, ten $3$-fold, six $5$-fold rotation axes ($62$ points: an
icosidodecahedron, a dodecahedron and an icosahedron). The character of these
singularities is described by the following proposition.

\begin{proposition}\label{localextr}
In the situation above, singular points of type I are minima (resp. maxima),
of type II maxima (resp. minima), and of type III saddle points for $H_{B}$
(resp. $\widetilde{H}_{B}$).
\end{proposition}

The proof of this fact is quite elementary. From Proposition \ref{criterion}%
.1) we deduce the character of singular points of type I. For type II it is
enough to use Proposition \ref{criterion}.2). For type III one have to
indicate two great circles such that the second derivatives along these curves
have different sign. As we will not use this fact in the sequel, we omit the details.

Hence the points of type I are the natural candidates for minimizing $H_{B}$
(resp. maximizing $\widetilde{H}_{B}$), and indeed, we will show in the next
section that they are global minimizers (resp. maximizers). However, if a
POVM\ is merely symmetric, the global extrema of entropy functions may also
occur in other (non inert) points. An example of this phenomenon can be found in
\cite{Ghietal03}, see also \cite{Bosetal12,ZozBosPor13}. Let us consider a symmetric (but
non-highly symmetric) POVM generated by the set of four Bloch vectors forming
a rectangle $B=\left\{  v_{1},-v_{1},v_{2},-v_{2}\right\}  $, where
$v_{1},v_{2}\in S^{2}$, $v_{1}\notin\left\{  -v_{2},v_{2}\right\}  $, and
$v_{1}\cdot v_{2}\neq0$, with $\operatorname{Sym}\left(  B\right)  \simeq
D_{2h}$ having three mutually perpendicular 2-fold rotation axis. In this way we get
six vectors in $S^{2}$ that are necessarily critical for $H_{B}$ and
$\widetilde{H}_{B}$: two perpendicular both to $v_{1}$ and to $v_{2}$, and
four lying in the plane generated by $v_{1}$ and $v_{2}$, proportional to $\pm
v_{1}\pm v_{2}$. The former are local maxima of $H_{B}$, and the latter either
local minima or saddle points, depending on the value of the parameter
$\alpha:=\arccos\left(  v_{1}\cdot v_{2}\right)  \in\left(  0,\pi\right)  $,
$\alpha \neq \pi /2$.
Let $\overline{\alpha}\approx1.17056$ be a unique solution of the equation
$(\cos(\overline{\alpha}/2))\ln(\tan^{2}(\overline{\alpha}/4))  =-2$ in the interval $\left(  0,\pi/2\right)  $. In
\cite{Ghietal03} the authors showed that for $\alpha\in\left(  0,\overline
{\alpha}\right]  $ the function $H_{B}$ (resp. $\widetilde{H}_{B}$) attains
the global minimum (resp. maximum) at the points $\pm\left(  v_{1}%
+v_{2}\right)  /\left(  2\left|  \cos\left(  \alpha/2\right)  \right|
\right)  $, whereas $\pm\left(  v_{1}-v_{2}\right)  /\left(  2\left|
\sin\left(  \alpha/2\right)  \right|  \right)  $ are saddle points, and for
$\alpha\in\left[  \pi-\overline{\alpha},\pi\right)  $ the situation is
reversed. However, for $\alpha\in\left(  \overline{\alpha},\pi-\overline
{\alpha}\right)  $ all these inert states become saddles, and two pairs of new
global minimizers emerge, lying symmetrically with respect to the old ones.
The appearance of this pitchfork bifurcation phenomenon shows also that, in
general, one cannot expect an analytic solution of the minimization problem in
a merely symmetric case. This is why we restrict our attention to highly
symmetric POVMs.

Note also that for highly symmetric POVMs we can use, instead of full symmetry
group $\operatorname*{Sym}\left(  B\right)$, the rotational symmetry group of $B$
that acts transitively on $B$, i.e., $C_{n}$ for the regular $n$-gon, $T$ for the
tetrahedron, $O$ for the cuboctahedron, cube and octahedron,  and $I$ for the
icosidodecahedron, dodecahedron and icosahedron.

\section{Global minima of entropy for highly symmetric POVMs in dimension two}
\label{Gloext}

\subsection{The minimization method based on the Hermite interpolation}
\label{HerInt}

In order to prove that the antipodal points to the Bloch vectors of POVM
elements are not only local but also global minimizers, we shall use a method
based on the Hermite interpolation.

Consider a sequence of points
$a\leq t_{1}<t_{2}<\ldots<t_{m}\leq b$, a sequence of positive integers
$k_{1},k_{2},\ldots,k_{m}$, and a real valued function
$f\in C^{D}([a,b])$, where $D:=k_{1}+k_{2}+\ldots+k_{m}$.
We are looking for a polynomial $p$ of degree less than $D$ that agree
with $f$ at $t_{i}$ up to a derivative of order $k_{i}-1$ (for $1\leq i\leq
m$), that is,
\begin{equation}
p^{(k)}(t_{i})=f^{(k)}(t_{i}),\qquad0\leq k<k_{i}\text{.} \label{Her}%
\end{equation}
The existence and uniqueness of such polynomial follows from the injectivity
(and hence also the surjectivity) of a linear map $\Phi:\mathbb{R}_{<D}\left[
X\right]  \rightarrow\mathbb{R}^{D}$ given by $\Phi\left(  p\right)
:=(p(t_{1}),p^{\prime}(t_{1}),\ldots,p^{(k_{1}-1)}(t_{1}),\ldots
,p(t_{m}),\ldots,p^{(k_{m}-1)}(t_{m}))$. We will also use the following
well-known formula for the error in Hermite interpolation \cite[Sec.~2.1.5]{StoBur02}:

\begin{lemma}
\label{hermite}For each $t\in\left(  a,b\right)  $ there exists $\xi\in\left(
a,b\right)  $ such that $\min\{t,t_{1}\}<\xi<\max\{t,t_{m}\}$ and%
\begin{equation}
\label{hererr}
f(t)-p(t)=\frac{f^{(D)}(\xi)}{D!}\prod_{i=1}^{m}(t-t_{i})^{k_{i}}\text{.}%
\end{equation}
\end{lemma}

Now we apply this general method in our situation. We will interpolate the function
$h:\left[  -1,1\right]  \rightarrow\mathbb{R}^{+}$ defined by (\ref{h}),
choosing the interpolation points from the set $T:=\{-gv\cdot v|g\in
G\}\subset\lbrack-1,1]$, where $v$ is the Bloch vector representation of the
fiducial vector and $-v$ is supposed to be the Bloch vector of a global
minimizer. We must distinguish two situations: either the inversion $-I\in G$
(equivalently $-v\in B$) or not. The former is the case for $G=D_{nh}$ (for
even $n$), $O_{h}$, $I_{h}$, and then $1\in T$, the latter for $G=D_{nh}$ (for
odd $n$), $T_{d}$, and then $1\notin T$. After reordering the elements of $T$
we obtain an increasing sequence $\{t_{i}\}_{i=1}^{m}$, where $m := \left|  T\right|$.
In particular, $t_{1}=-1$. We are looking for a polynomial $p_{v}$ that matches the
values of $h$ at all points from $T$ and the values of $h^{\prime}$ at all points but
$-1$ and, possibly, $1$, if $1\in T$, i.e. such that (\ref{Her}) holds for
$f=h$ with%
\begin{equation}
k_{i}:=\left\{
\begin{tabular}
[c]{ll}%
$1$, & $\text{\text{if}\ }t_{i}\in\left\{  {-1,1}\right\}  $\\
$2$, & $\text{\text{otherwise}}$%
\end{tabular}
\right.  \text{.}%
\end{equation}
Then $\deg p_{v}<D(v):=2m-2$, if $1\in T$, and $\deg p_{v}<D(v):=2m-1$, otherwise.
Though $h$ is not differentiable at $t_1$, we still can use (\ref{hererr})
to estimate the interpolation error, as proof of Lemma~\ref{hermite} is based
on repeated usage of Rolle's Theorem.

If $1\in T$, then $t_{m}=1$ and we have
\begin{equation}
\prod_{i=1}^{m}(t-t_{i})^{k_{i}}=(t+1)(t-1)\prod_{i=2}^{m-1}(t-t_{i})^{2}%
\leq0\text{,}%
\end{equation}
for $t\in\left[  -1,1\right]  $. Similarly, if $1\notin T$, then $t_{m}<1$ and%
\begin{equation}
\prod_{i=1}^{m}(t-t_{i})^{k_{i}}=(t+1)\prod_{i=2}^{m}(t-t_{i})^{2}\geq0
\end{equation}
for $t\in\left[  -1,1\right]  $. Moreover, inequalities above turn into
equalities only for $t\in T$. Furthermore, as all the derivatives of $h$ of
even order are strictly negative in $\left(  -1,1\right)  $ and these of odd
order greater than $1$ are strictly positive, we get
\begin{equation}
h^{(D\left(  v\right)  )}(\xi)=\left\{
\begin{tabular}
[c]{ll}%
$h^{(2m-2)}(\xi)<0$, & if $1\in T$\\
$h^{(2m-1)}(\xi)>0$, & if $1\notin T$%
\end{tabular}
\right.
\end{equation}
for each $\xi\in\left(  -1,1\right)  $. Hence and from Lemma \ref{hermite} the
interpolating polynomial $p_{v}$ fulfills $p_{v}(t)=h(t)$ if and only if $t\in
T$ and it interpolates $h$ from below, see the illustration of this for the
octahedral POVM in Fig.~\ref{F1}.

\begin{figure}[h]
\begin{center}
\includegraphics[scale=0.7]{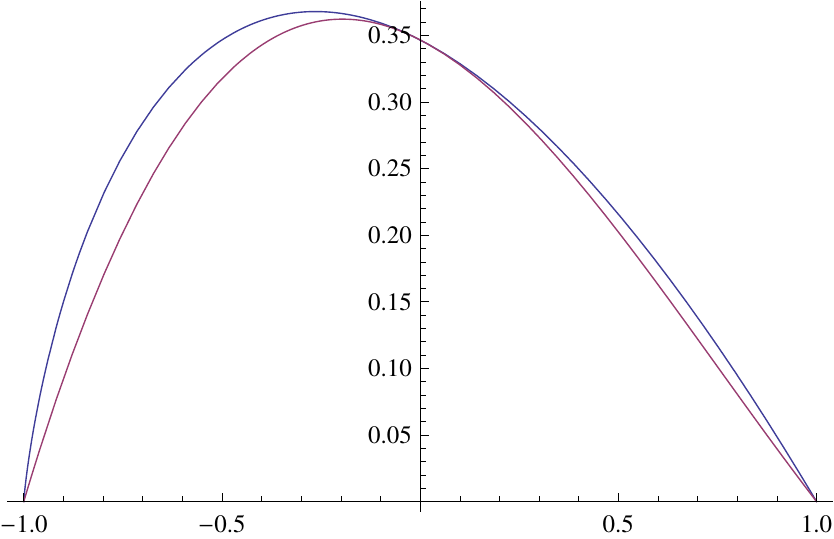}
\end{center}
\caption{The cubic polynomial function $p_{v}$ (\emph{purple}) interpolating $h$
(\emph{violet}) from below for the octahedral measurement, with
$t_{1}=-1$, $t_{2}=0$ and $t_{3}=1$.}%
\label{F1}
\end{figure}

Let us define now a $G$-invariant polynomial function $P_{v}:\mathbb{R}%
^{3}\rightarrow\mathbb{R}$ replacing $h$ in (\ref{entfor}) by its
interpolation polynomial $p_{v}$, i.e.,%
\begin{equation}
P_{v}(u):=\ln\frac{|Gv|}{2}+\frac{2}{|G|}\sum_{g\in G}p_{v}(gv\cdot u)
\end{equation}
for $u\in\mathbb{R}^{3}$. Combining the above facts, we get
\begin{equation}\label{interP}
H_{v}(u)=\ln\frac{|Gv|}{2}+\frac{2}{|G|}\sum_{g\in G}h(gv\cdot u)\geq
\ln\frac{|Gv|}{2}+\frac{2}{|G|}\sum_{g\in G}p_{v}(gv\cdot u)=P_{v}(u)
\end{equation}
for $u\in S^{2}$, and $H_{v}(-gv)=P_{v}(-gv)$ for $g\in G$.

In consequence, now it is enough to show that $-v$ is a global minimizer of $P$
(and hence all the elements of its orbit $\left\{  -gv:g\in G\right\}  $ are),
because then we have $H_{v}(u)\geq P_{v}(u)\geq P_{v}(-v)=H_{v}(-v)$ for all
$u\in S^{2}$. This method of finding global minima was inspired by the one used
in \cite{Oreetal11} for $G=T_{d}$, where $h$ is constant. Note, however, that
a similar technique was used by Cohn, Kumar and Woo \cite{CohCum07,CohWoo12}
to solve the problem of potential energy minimization on the unit sphere. The
whole idea can be traced back even further to \cite{Yud92} and \cite{And97}.
Our method has been already adapted in \cite{Arn15} to provide some upper bounds
for the informational power of $t$-design POVMs for $1 \leq t \leq 5$.

Of course, the lower the degree of the interpolating polynomial $p_{v}$ is,
the easier it is to find the minima of $P_{v}$, as $\deg P_{v}\leq\deg p_{v}$.
The last quantity in turn depends on the cardinality of $T:=\{-gv\cdot
v|g\in G\}$, that can be calculated by analyzing double cosets of isotropy
subgroups of any subgroup $K\subset G\cap SO(3)$ acting transitively on $B$,
because $T=\{-gv\cdot v|g\in K\}$ and for $h,g\in K$, if $h$ is in a double
coset $K_{v}gK_{v}$ or $K_{v}g^{-1}K_{v}$, then $hv\cdot v=gv\cdot v$. Hence
$\left|  T\right|  \leq n(v):=n_{s}(v)+\frac{1}{2}n_{a}(v)$, where $n_{s}(v)$
is the number of self-inverse double cosets of $K_{v}$, i.e.\ the cosets
fulfilling $K_{v}gK_{v}=K_{v}g^{-1}K_{v}$, and $n_{a}(v)$ is the number of non
self-inverse ones. Thus
\begin{equation}
\deg p_{v}\leq\left\{
\begin{tabular}
[c]{ll}%
$2n(v)-3$, & if $-v\in Kv$\\
$2n(v)-2$, & if $-v\notin Kv$%
\end{tabular}
\right.  \text{.} \label{degree}%
\end{equation}
Moreover, for $g\in K$, using the well-known formula for the cardinality of a
double coset, see, e.g.\ \cite[Prop.~5.1.3]{Bog08}, we have $\left|
K_{v}gK_{v}\right|  =\left|  K_{v}\right|  \left|  K_{v}/\left(  K_{v}\cap
K_{gv}\right)  \right|  =\left|  K_{v}\right|  $, if $gv=v$ or $gv=-v$, and
$\left|  K_{v}\right|  ^{2}$, otherwise. Hence, if $-v\in Kv$, then $\left|
Kv\right|  \left|  K_{v}\right|  =\left|  K\right|  =2\left|  K_{v}\right|
+\left(  n_{s}(v)-2\right)  \left|  K_{v}\right|  ^{2}+n_{a}(v)\left|
K_{v}\right|  ^{2}$, and so $n_{s}(v)+n_{a}(v)=\left(  \left|  Kv\right|
-2\right)  /\left|  K_{v}\right|  +2$. Analogously, if $-v\notin Kv$, then
we have $n_{s}(v)+n_{a}(v)=\left(  \left|  Kv\right|  -1\right)  /\left|
K_{v}\right|  +1$. Using these facts and (\ref{degree}) we get finally%
\begin{equation}
\deg p_{v}\leq\left\{
\begin{tabular}
[c]{ll}%
$\frac{\left|  Kv\right|  -2}{\left|  K_{v}\right|  }+n_{s}(v)-1$, & if $-v\in
Kv$ \\ \noalign{\vskip 1mm}
$\frac{\left|  Kv\right|  -1}{\left|  K_{v}\right|  }+n_{s}(v)-1$, & if
$-v\notin Kv$%
\end{tabular}
\right.  \text{.} \label{degreebound}%
\end{equation}
Applying (\ref{degreebound}) to HS-POVMs in dimension two we get the
upper bounds for the degree of interpolating polynomials gathered in Tab.~\ref{upperbounds}.

\begin{table}[h]
\caption{HS-POVMs in dimension two: upper bounds for the number of interpolating points ($n(v)$) and the degree of interpolating polynomial.}
\label{upperbounds}
\centering
\begin{tabular}{cccccccc}
\hline\noalign{\smallskip}
$Kv$ & $\left|  Kv\right|  $ & $K$ & $K_{v}$ & $n_{a}(v)$ & $n_{s}(v)$  & $n(v)$ & $\deg
p_{v}\leq$\\
\noalign{\smallskip}\hline\noalign{\smallskip}
regular $n$-gon ($n$-even) & $n$ & $C_{n}$ & $C_{1}$ & $n-2$ & $2$  & $n/2+1$ &
$n-1$\\
regular $n$-gon ($n$-odd) & $n$ & $C_{n}$ & $C_{1}$ & $n-1$ & $1$  & $n/2+1/2$ & $n-1$\\
tetrahedron & $4$ & $T$ & $C_{3}$ & $0$ & $2$ & $2$ & $2$\\
octahedron & $6$ & $O$ & $C_{4}$ & $0$ & $3$ & $3$ & $3$\\
cube & $8$ & $O$ & $C_{3}$ & $0$ & $4$  & $4$ & $5$\\
cuboctahedron & $12$ & $O$ & $C_{2}$ & $4$ & $3$  & $5$ & $7$\\
icosahedron & $12$ & $I$ & $C_{5}$ & $0$ & $4$  & $4$ & $5$\\
dodecahedron & $20$ & $I$ & $C_{3}$ & $4$ & $4$  & $6$ & $9$\\
icosidodecahedron & $30$ & $I$ & $C_{2}$ & $14$ & $2$  & $9$ & $15$\\
\noalign{\smallskip}\hline
\end{tabular}
\end{table}

To find global minimizers of $P_{v}$ we can express the polynomial in terms of
primary and secondary invariants for the corresponding ring of $G$-invariant
polynomials. In fact, as we will see in the next section, only the former will
be used.

\subsection{Group invariant polynomials}\label{invpol}

The material of this subsection is taken from \cite[Ch.\ 3]{DerKem02} and
\cite{JarMicSha84}, see also \cite{FriHau10}.
Let $G$ be a finite subgroup of the general linear group
$GL_{n}\left(  \mathbb{R}\right)  $. By $\mathbb{R}\left[  x_{1},\ldots
,x_{n}\right]  ^{G}$ we denote the ring of $G$-invariant real polynomials in
$n$ variables. Its properties were studied by Hilbert and Noether at the
beginning of twentieth century. In particular, they showed that $\mathbb{R}%
\left[  x_{1},\ldots,x_{n}\right]  ^{G}$ is finitely generated as an
$\mathbb{R}$-algebra. Later, it was proven that it is possible to represent
each $G$-invariant polynomial in the form ${\textstyle\sum_{j=1}^{m}}
P_{j}\left(  \theta_{1},\ldots,\theta_{n}\right)  \eta_{j}$, where $\theta
_{1},\ldots,\theta_{n}$ are algebraically independent homogeneous $G$-invariant
polynomials called \textsl{primary invariants}, forming so called
\textsl{homogeneous system of parameters}, $\eta_{1}=1,\ldots,\eta_{m}$ are
$G$-invariant homogeneous polynomials called \textsl{secondary invariants}, and
$P_{j}$ ($j=1,\ldots,m$) are elements from $\mathbb{R}\left[  x_{1}%
,\ldots,x_{n}\right]  $. Moreover, $\eta_{1},\ldots,\eta_{m}$ can be chosen in
such a way that they generate  $\mathbb{R}\left[  x_{1},\ldots,x_{n}\right]
^{G}$ as a free module over $\mathbb{R}\left[  \theta_{1},\ldots,\theta
_{n}\right]  $. Both sets of polynomials combined form so called
\textsl{integrity basis}. Note that neither primary nor secondary invariants
are uniquely determined. If $m=1$, we call the basis \textsl{regular} and
the group $G$ \textsl{coregular}. The invariant polynomial functions on
$\mathbb{R}^{n}$ separate the $G$-orbits. In consequence, the map
$\mathbb{R}^{n}/G\ni Gx\rightarrow\left(  \theta_{1}\left(  x\right)
,\ldots,\theta_{n}\left(  x\right), \eta_{2}\left(  x\right)  ,\ldots
,\eta_{m}\left(  x\right)  \right)  \in\mathbb{R}^{n+m-1}$ maps bijectively
the orbit space onto an $n$-dimensional connected closed semialgebraic subset
of $\mathbb{R}^{n+m-1}$. There is also a correspondence between the orbit
stratification of $\mathbb{R}^{n}/G$ and the natural stratification of this
semi-algebraic set into the primary strata, i.e.\ connected semialgebraic
differentiable varieties. If $G\subset O\left(  n\right)  $ is a coregular
group acting irreducibly on $\mathbb{R}^{n}$, we may assume that $\theta
_{1}\left(  x\right)  =\sum_{i=1}^{n}x_{i}^{2}$ is a non-constant invariant
polynomial of the lowest degree. Then the \textsl{orbit map} $\omega
:S^{n-1}/G\ni\nolinebreak Gx\rightarrow\left(  \theta_{2}\left(  x\right)  ,\ldots
,\theta_{n}\left(  x\right)  \right)  \in\mathbb{R}^{n-1}$ is also one-to-one
and its range is a semialgebraic $\left(  n-1\right)  $-dimensional set. In
consequence, the minimizing of  a $G$-invariant polynomial $P\left(
x_{1},\ldots,x_{n}\right)  $ on $S^{n-1}$ is equivalent to the minimizing of
the respective polynomial $P_{1}\left(  \theta_{1},\ldots,\theta_{n}\right)  $
on the range of $\omega$. In the 1980s Abud and Sartori proposed a general
procedure for finding the algebraic equations and inequalities defining this
set and its strata, and thus also a general scheme for finding minima of
$P_{1}$ on the range of the orbit map, see \cite{SarVal96,SarVal03}.

Let us now take a closer look at the $G$-invariant
polynomials of three real variables, which will be of our
particular interest while considering HS-POVMs in dimension two.
An element from $GL_{n}\left(  \mathbb{R}\right)  $ is called a
\textsl{pseudo-reflection}, if its fixed-points space has codimension one. The
classical Chevalley-Shephard-Todd theorem says that every
\textsl{pseudo-reflection }(i.e.\ generated by pseudo-reflections)
\textsl{group} is coregular. As all the symmetry groups of polyhedra
representing HS-POVMs\ in dimension two ($D_{nh}$, $T_{d}$, $O_{h}$, $I_{h}$)
are pseudo-reflection groups, the interpolating polynomials can be expressed
by their primary invariants listed below. Put $\rho:=x^{2}+y^{2}$, $\gamma
_{n}:=\Re\left(  x+iy\right)  ^{n}$, $I_{2}:=x^{2}+y^{2}+z^{2}$, $I_{3}:=xyz$,
$I_{4}:=x^{4}+y^{4}+z^{4}$, $I_{6}:=x^{6}+y^{6}+z^{6}$, $I_{6}^{\prime}
:=(\tau^{2}x^{2}-y^{2})(\tau^{2}y^{2}-z^{2})(\tau^{2}z^{2}-x^{2})$ and
$I_{10}:=(x+y+z)(x-y-z)(y-z-x) (z-y-x)(\tau^{-2}x^{2}-\tau^{2}y^{2})(\tau
^{-2}y^{2}-\tau^{2}z^{2})(\tau^{-2}z^{2}-\tau^{2}x^{2})$, where $\tau
:=(1+\sqrt{5})/2$ (the \emph{golden ratio}). Note that the indices coincide with the
degrees of invariant polynomials. Then (notation and results are taken from
\cite{JarMicSha84}) for the canonical representations of these groups, i.e.\
if coordinates $x$, $y$ and $z$ are so chosen that the origin is the fixed
point for the group action and: the $x$ and $z$ axes are $2$- and $n$-fold
axes, respectively ($D_{nh}$); the $3$-fold axes pass through vertices of a
tetrahedron at $(1,1,1)$, $(1,-1,-1)$, $(-1,1,-1)$, $(-1,-1,1)$ ($T_{d}$);
$x$, $y$ and $z$ axes are the $4$-fold axes ($O_{h}$); the $5$-fold axes pass
through the vertices of an icosahedron at $(\pm\tau,\pm1,0)$, $0,\pm\tau
,\pm1)$, $(\pm1,0,\pm\tau)$ ($I_h$), we get the primary invariants listed in Tab.~\ref{priminv}:

\begin{table}[h]
\caption{Primary invariants for four point groups.}
\label{priminv}
\centering
\begin{tabular}{cc}
\hline\noalign{\smallskip}
group & primary invariants  \\
\noalign{\smallskip}\hline\noalign{\smallskip}
$D_{nh}$ & $z^{2},\rho,\gamma_{n}$ \\
$T_{d}$ & $I_{2},I_{3},I_{4}$ \\
$O_{h}$ & $I_{2},I_{4},I_{6}$ \\
$I_{h}$ & $I_{2},I_{6}^{\prime},I_{10}$ \\
\noalign{\smallskip}\hline
\end{tabular}
\end{table}

In \cite{JarMicSha84} the stratification of the range of the orbit map is
analytically described in all these cases.

\subsection{The main theorem}\label{mainhs2}

\begin{theorem}
\label{main}For HS-POVMs in dimension two the points lying on the orbit of
the point antipodal to the Bloch vector of the fiducial vector (that is the
Bloch vector of the state orthogonal to the fiducial vector) are the only global
minimizers (resp. maximizers) for the entropy of measurement (resp. the
relative entropy of measurement).
\end{theorem}

\begin{proof} We will give a proof of the theorem in two steps. Firstly, we show that the antipodal points to the
Bloch vectors of POVM elements, i.e.\ the points $\left\{  -gv:g\in G\right\}
$ are the global minima of the $G$-invariant polynomial $P_{v}$ constructed in
Sect.~\ref{HerInt}. (In particular, this is true if $P_v$ is constant.)
Then we prove the uniqueness of designated global minimizers of the entropy of measurement.

 We shall use the a priori estimates for $\deg P_{v}$ that
can be read from Tab.~\ref{upperbounds} and the primary invariants of $G$ listed in Tab.~\ref{priminv}. We
may exclude the trivial case when the HS-POVM in question is PVM represented
by two antipodal points on the Bloch sphere (digon), as in this situation the
minimal value of $H$ equal $0$ is achieved at these points and the assertion follows.
The proof is divided into four cases according to the symmetry group of the HS-POVM.
\medskip

\textsl{Case I (prismatic symmetry)}\smallskip

\textbf{Regular }$n$\textbf{-gon}. In Sect.~\ref{EntropyBloch} we showed that in this
case it is enough to look for the global minimizers on the circle
$S^{1}:=\left\{  (x,y)\in\mathbb{R}^{2}:x^{2}+y^{2}=1\right\}  $ containing
the $n$-gon. Its symmetry group acts on the plane $z=0$ as the dihedral group
$D_{n}$, and so the interpolating polynomial $P_{v}$ restricted to the circle
$S^{1}$ can be expressed in terms of its primary invariants, i.e.\ $\rho
=x^{2}+y^{2}$ and $\gamma_{n}=\Re(x+iy)^{n}$. Since $\deg P_{v}<n$, it follows
that $P_{v}|_{S}$ has to be a linear combination of $\rho^{m}$, $0\leq2m<n$,
and hence constant.\medskip

\textsl{Case II (tetrahedral symmetry)}\smallskip

\textbf{Tetrahedron}. This case is immediate, as $\deg P_{v}\leq\deg p_{v}%
\leq2$, and so $P_{v}$ has to be constant on the sphere $S^{2}$.\medskip

\textsl{Case III (octahedral symmetry)}\smallskip

For $O_{h}$ we have inert states at the $O_{h}$-orbits of the points:
$x_{1}:=(0,0,1)$ (vertices of an octahedron), $x_{2}:=\frac{1}{\sqrt{2}%
}(0,1,1)$ (vertices of a cuboctahedron), and $x_{3}:=\frac{1}{\sqrt{3}%
}(1,1,1)$ (vertices of a cube). Using the Lagrange multipliers it is easy to
check that these points are the only critical points for $I_{4}$ and $I_{6}$
restricted to the sphere $S^{2}$. By comparing the values of $I_{4}$ and
$I_{6}$ (which are $I_{4}\left(  x_{1}\right)  =1$, $I_{4}\left(
x_{2}\right)  =1/2$, $I_{4}\left(  x_{3}\right)  =1/3$, $I_{6}\left(
x_{1}\right)  =1$, $I_{6}\left(  x_{2}\right)  =1/4$, $I_{6}\left(
x_{3}\right)  =1/9$), we find that the points lying on the orbit of $x_{3}$
are global minimizers both for $I_{4}$ and $I_{6}$.

\smallskip
\textbf{Octahedron}. This case is straightforward, as for $v=x_{1}$ we have
$\deg P_{v}\leq\deg p_{v}\leq3$, and so $P_{v}$ has to be constant on the
sphere $S^{2}$.

\smallskip
\textbf{Cube}. In this case we have $v=x_{3}$ and $\deg P_{v}\leq\deg
p_{v}\leq5$. In consequence, $P_{v}$ must be a linear combination of $1$,
$I_{2}$, $I_{4}$, and $I_{2}^{2}$. After the restriction to the sphere,
$P_{v}|_{^{S^{2}}}$ can be expressed as $A+BI_{4}$, for some $A,B\in
\mathbb{R}$. Thus, all we need to know now is the sign of $B$. Calculating the
values of $P_{v}$ in two points from different orbits (e.g.\ $x_{1}$ and
$x_{3}$) and solving the system of two linear equations we get $B=(3/8)\ln
(27/16)>0$. Thus the global minimizers for $P_{v}$ are the same as for $I_{4}%
$, i.e.\ they lie on the orbit of $v$ or, equivalently, $-v$, as required.

\smallskip
\textbf{Cuboctahedron}. For the cuboctahedral measurement we have $v=x_{2}$
and $\deg P_{v}\leq\deg p_{v}\leq7$. Consequently, $\deg P_{v}\leq6$ and
$P_{v}$ is a linear combination of $1$, $I_{2}$, $I_{4}$, $I_{2}^{2}$, $I_{6}%
$, $I_{4}I_{2}$, and $I_{2}^{3}$. Hence, after the restriction to the sphere
$S^{2}$, we get $P_{v}|_{^{S^{2}}}=A+BI_{4}+CI_{6}$, for some $A,B,C\in
\mathbb{R}$. Put $\beta:=-B/(3C)$. Clearly, all inert states are critical for
$P_{v}|_{^{S^{2}}}$ with $P_{v}\left(  x_{1}\right)  =A+C(1-3\beta)$,
$P_{v}\left(  x_{2}\right)  =A+C(1-6\beta)/4$, $P_{v}\left(  x_{3}\right)
=A+C(1-9\beta)/9$. One can show easily that they are only critical points
unless $1/4<\beta<1/2$. In this case there are another critical points,
namely the orbit of the point $x_{4}:=\left(  \sqrt{4\beta-1},\sqrt{1-2\beta
},\sqrt{1-2\beta}\right)  $ with $P_{v}\left(  x_{4}\right)  =C\left(
1-9\beta+24\beta^{2}-24\beta^{3}\right)  $. To find $B$ and $C$, we need to
calculate the values of $P_{v}$ in three points from different orbits (e.g.\
$x_{1}$, $x_{2}$ and $x_{3}$) and to solve the system of three linear
equations. In this way we get $B=\frac{520}{9}\ln2-37\ln3<0$, $C=-\frac
{364}{9}\ln2+26\ln3>0$ and $\beta\approx0.3775$. Comparison of the values that
$P_{v}$ achieves at points $x_{1}$, $x_{2}$, $x_{3}$ and $x_{4}$ leads to the
conclusion that the global minima are achieved for the vertices of
cuboctahedron that form the orbit of $v$ and thus also of $-v$.\medskip

\textsl{Case IV (icosahedral symmetry)}\smallskip

The inert states for $I_{h}$, that is the $I_{h}$-orbits of points:
$x_{1}=(0,0,1)$ (vertices of an icosidodecahedron), $x_{5}:=\frac{1}%
{\sqrt{\tau+2}}(0,\tau,1)$ (vertices of an icosahedron), and $x_{6}:=\frac
{1}{\sqrt{3}}(0,\frac{1}{\tau},\tau)$ (vertices of a dodecahedron) are the
only critical points for $I_{6}^{\prime}$. They are, correspondingly, saddle,
minimum, and maximum points with values: $0$, $-(2+\sqrt{5})/5$, and
$(2+\sqrt{5})/27$, respectively. For $I_{10}$, the $I_{h}$-orbit of $x_{6}$
also coincides with the set of the global maxima, and we have local maxima at
the $I_{h}$-orbit of $x_{5}$ and saddle points at the orbit of $x_{1}$, but
there are also non-inert critical points, namely sixty minima at the vertices
of a non-Archimedean vertex truncated icosahedron (Fig.~14 in \cite{ZefArd07}),
and sixty saddles at the vertices of an edge truncated Archimedean vertex truncated
icosahedron (Fig.~5 in \cite{Zef11}), see \cite[p.~26]{JarMicSha84}.

\smallskip
\textbf{Icosahedron}. This case is immediate, as $v=x_{5}$ and $\deg P_{v}%
\leq\deg p_{v}\leq5$. Hence $P_{v}$ restricted to $S^{2}$ is constant.

\smallskip
\textbf{Dodecahedron}. In this case $v=x_{6}$ and $\deg P_{v}\leq\deg
p_{v}\leq9$. Therefore $P_{v}$ must be a linear combination of $1$, $I_{2}$,
$I_{2}^{2}$, $I_{2}^{3}$, $I_{6}^{\prime}$, $I_{2}^{4}$ and $I_{6}^{\prime
}I_{2}$. After restriction to $S^{2}$ we obtain $P_{v}|_{^{S^{2}}}%
=A+BI_{6}^{\prime}$, for some $A,B\in\mathbb{R}$. We can calculate $B$ using
the same method as in the cubical case. As it turns out to be negative
($B\approx-0.06509$), the global minimizers coincide with the global
maximizers for $I_{6}^{\prime}$, i.e.\ they are the vertices of the dodecahedron.

\smallskip
\textbf{Icosidodecahedron}. The icosidodecahedral case ($v=x_{1}$) is the most
complicated one. Since $\deg P_{v}\leq\deg p_{v}\leq15$, and $P_{v}$ must be a
linear combination of polynomials $1$, $I_{2}$, $I_{2}^{2}$, $I_{2}^{3}$,
$I_{6}^{\prime}$, $I_{2}^{4}$, $I_{6}^{\prime}I_{2}$, $I_{2}^{5}$,
$I_{6}^{\prime}I_{2}^{2}$, $I_{10}$, $I_{2}^{6}$, $I_{6}^{\prime}I_{2}^{3}$,
$(I_{6}^{\prime})^{2}$, $I_{10}I_{2}$, $I_{2}^{7}$, $I_{6}^{\prime}I_{2}^{4}$,
$(I_{6}^{\prime})^{2}I_{2}$, and $I_{10}I_{2}^{2}$. Restriction to $S^{2}$
gives us: $P_{v}|_{^{S^{2}}}=A+BI_{6}^{\prime}+CI_{10}+D(I_{6}^{\prime})^{2}$,
for some $A,B,C,D\in\mathbb{R}$. Both of the polynomials $I_{6}^{\prime}$ and
$I_{10}$ take the value $0$ at $x_{1}$, which is obviously a critical point
for $P_{v}|_{^{S^{2}}}$. As we have conjectured that the vertices of the
icosidodecahedron are the global minimizers of $P_{v}|_{^{S^{2}}}$, it is
enough to prove that $\tilde{P}:=P_{v}|_{^{S^{2}}}-A$ is nonnegative. We keep
proceeding like in the previous cases to obtain formulae for $B$, $C$, and $D$:
\begin{align*}
B = & -(1/50) (-2 + \sqrt5) (7122 \sqrt5 \arcoth(\sqrt5) +
   3 (-3728 + 2773 \sqrt5) \ln2 \\&+
    39575 \ln 3 - 4700 \ln 5 -
   8319 \sqrt 5 \ln(7 + 3 \sqrt 5)),\\
C = & \ (1/180) (-108414 \arcoth(3/\sqrt5) + 47970 \arcoth(\sqrt5) +
   \sqrt5 (-16352 \ln2\\& + 51120 \ln3 - 5265 \ln5)),\\
D = & \ (29/900) (9 - 4 \sqrt5) (53766 \sqrt5 \arcoth(3/\sqrt5) -
   23418 \sqrt5 \arcoth(\sqrt5)\\& + 34816 \ln2 - 126450 \ln3 +
   15075 \ln5).
\end{align*}
The range $\Omega$ of the orbit map $\omega:S^{2}/I_{h}\ni I_{h}w\rightarrow\left(
I_{6}^{\prime}\left(  w\right)  ,I_{10}\left(  w\right)  \right)
\in\mathbb{R}^{2}$ is the curvilinear triangle (see Fig.~\ref{F2}) defined by the following
inequalities imposed on the coordinates
$\left( \theta_{1}, \theta_{2} \right) \in \mathbb{R}^{2}$:
\begin{align*}
-\frac{2\tau+1}{5}\leq\theta_{1}\leq\frac{2\tau+1}{27}\text{, \ \ }%
(7-4\tau)\theta_{1}\leq\theta_{2}\text{,}%
\end{align*}
\begin{align*}
0  &  \leq J_{15}^{2}:=4\theta_{1}^{2}-8(3+4\tau)\theta_{1}\theta
_{2}-91(3-2\tau)\theta_{1}^{3}+4(5+8\tau)\theta_{2}^{2}+\\&  + 159(1-2\tau)\theta_{1}^{2}\theta_{2}+688(13-8\tau)\theta_{1}^{4}%
+325(1+2\tau)\theta_{1}\theta_{2}^{2}+ \\&  -720(7-4\tau)\theta_{1}^{3}\theta_{2}-1728(55-34\tau)\theta_{1}%
^{5}-25(11+18\tau)\theta_{2}^{3}\text{,}%
\end{align*}
where $J_{15}$ is the only secondary invariant for the icosahedral
group $I$ \cite[Tab.~IIIb]{JarMicSha84}.

\begin{figure}[htb]
\begin{center}
\includegraphics[scale=0.4]{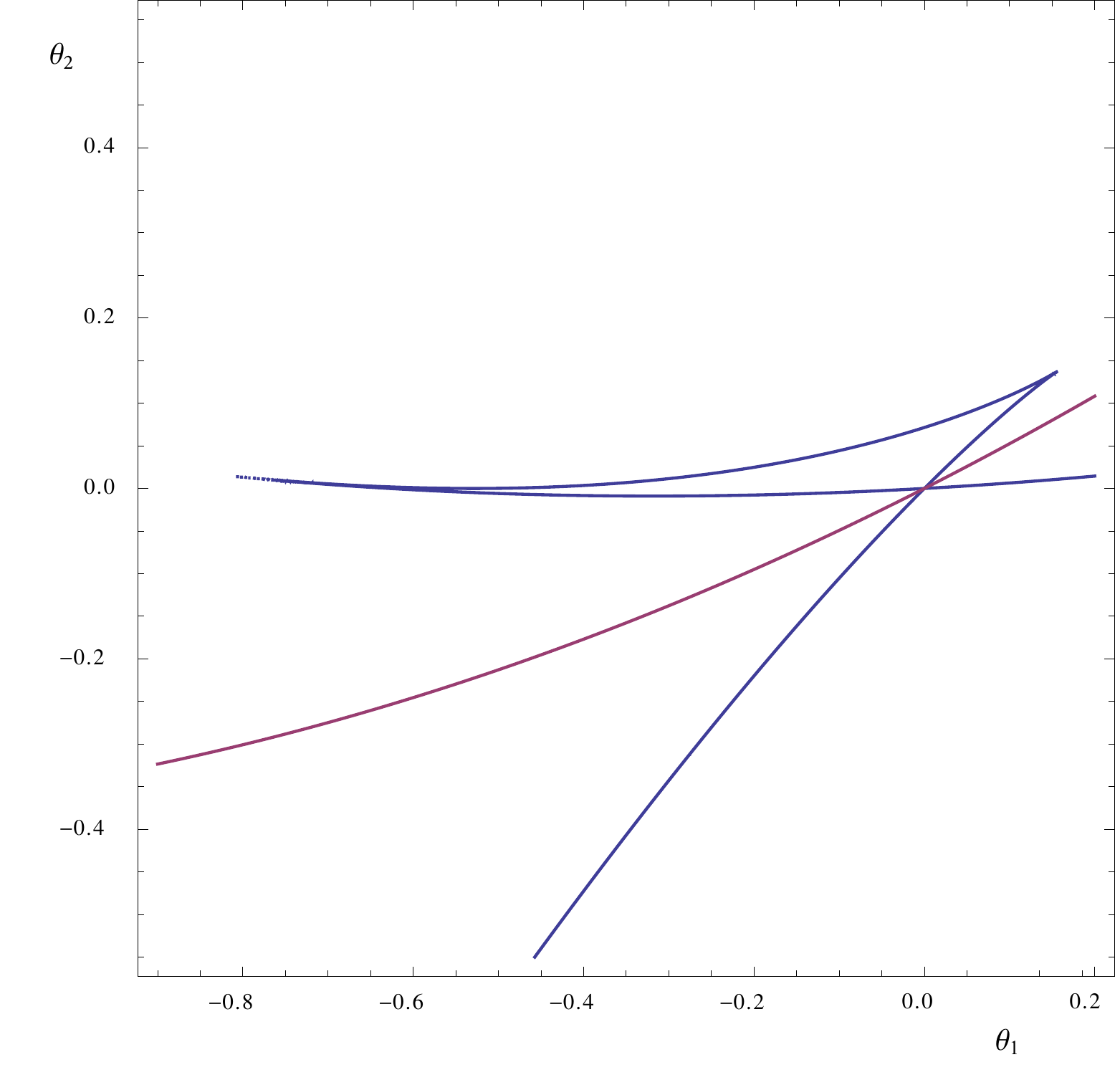}
\end{center}
\caption[The zero level set for $P_{1}$ and for $J_{15}^{2}$.]{The zero level set for $P_{1}$ (\emph{purple}) and for $J_{15}^{2}$
(\emph{violet}).}
\label{F2}
\end{figure}

Define $P_{1}\left(  \theta_{1},\theta_{2}\right)  :=B\theta_{1}+C\theta
_{2}+D\theta_{1}^{2}$ for $\left(  \theta_{1},\theta_{2}\right)  \in\Omega$.
Then $\tilde{P}\left(  w\right)  =\linebreak P_{1}\left(  \omega\left(
\left(  I_{h}\right)  w\right)  \right)  $ for $w\in S^{2}$. The level sets of
$P_{1}$ are parabolas and the zero level parabola given by $\theta
_{2}=-\left(  B/C\right)  \theta_{1}-(D/C)\theta_{1}^{2}$ (the purple curve in
Fig.~\ref{F2}) divides the plane into two regions: $\left\{  P_{1} \geq 0\right\}  $ and
$\left\{  P_{1}<0\right\}  $. Now, it is enough to show that the zero level
set of $P_{1}$ meets with the zero level set of $J_{15}^{2}$ (the violet curve
in Fig.~\ref{F2}), which defines the boundary of $\Omega$ only at $\left(  \theta
_{1},\theta_{2}\right)  =(0,0)$, since in this case $P_{1}$ has the same sign
over the whole $\Omega$, and, in consequence, $\tilde{P}$ is positive on the
whole unit sphere. This approach reduces the complexity of the problem by
lowering the degree of a polynomial equation to be solved. In fact, now it is
enough to show that the polynomial $Q(\theta_1):=J_{15}^{2}\left(  \theta_{1},-\left(
B/C\right)  \theta_{1}-(D/C)\theta_{1}^{2}\right)  /\theta_{1}^{2}$ of degree
$4$ has no real roots. This can be done in a standard way by using Sturm's theorem, the method which we recall briefly below.

The \emph{Sturm chain} for polynomial $q$ is a sequence $q_0, q_1,\ldots, q_m$, where $q_0=q$, $q_1=q'$,
$q_i =-\textrm{rem}(q_{i-2},q_{i-1})$ for $i=2,\ldots,m$, and $m\leq\deg q$ is the minimal number $i$
such that $\textrm{rem}(q_{i-1},q_{i})=0$ (by $\textrm{rem}(r,s)$ we denote the reminder of division of $r$ by $s$).
Sturm's theorem states that the number of roots of $q$ in $(a,b)$ for $-\infty\leq a<b\leq+\infty$ equals to
the difference between the numbers of sign changes in the Sturm chain for $q$ evaluated in $b$ and $a$
(for more details see, e.g.\ \cite[Sect.~2.2]{BasPolRo06}). Thus, to finish the proof for icosidodecahedron,
we calculate Sturm's chain for $Q$, evaluate it at $\pm\infty$ and show that  numbers of sign changes
do not differ.\footnote{To determine the signs of this expression the \emph{Mathematica} command \texttt{Sign} has been used.}

We end the proof with showing that there are no other (global) minimizers of the entropy.

It follows from (\ref{interP}) that if $w\in S^{2}$ is a
global minimizer for $H_{v}$, then it is also a global minimizer for $P_{v}$,
since $P_{v}(-v)\leq P_{v}(w)\leq H_{v}(w)=H_{v}(-v)=P_{v}(-v)$. The same
argument gives us $h(w\cdot u)=p(w\cdot u)$ for every $u\in Gv$, and so
$\{w\cdot u:u\in Gv\}\subset T=\{-v\cdot u:u\in Gv\}$.

Put $a_{u}:=w\cdot u$ for $u\in Gv$ and $k:=\left|  Gv\right|  $. Now it is
enough to show that $-1\in T_{w}:=\{a_{u}:u\in Gv\}$, since then $w\in
G\left(  -v\right)  $. We know that $\sum_{u\in Gv}a_{u}=0$.
For informationally-complete HS-POVMs we have additionally $\sum_{u\in
Gv}a_{u}^{2}=k/3$ (as $Gv$ is $2$-design), and, for icosahedral group,
$\sum_{u\in Gv}a_{u}^{4}=k/5$ (as $Gv$ is\linebreak $4$-design). Moreover, $1\in T_{w}$
implies $-1\in T_{w}$ for octahedral and icosahedral group. Using all these
facts and the form of the interpolating set for respective informationally
complete HS-POVMs (see Tab.~\ref{intersets}) we see that in all seven cases the assumption
$-1\notin T_{w}$ leads to the immediate contradiction. On the other hand, for
a regular polygon $w$ must lie on the circle containing this polygon (see Sect.~\ref{EntropyBloch}).
Then $w\cdot\left(  -v\right)  \in T$, implies $w\in G\left(  -v\right)
$, as desired.\qedhere

\begin{center}
\begin{table}[h]
\caption{The interpolating sets for HS-POVMs in dimension two.}
\label{intersets}
\centering
\begin{tabular}{ccc}
\hline\noalign{\smallskip}
$Gv$ & $\left|  Gv\right|  $ & $T$\\
\noalign{\smallskip}\hline\noalign{\smallskip}
regular $n$-gon ($n$ even) & $n$ & $\big\{\cos\left(  \frac{2\pi j}{n}\right)
:j=1,\ldots,n\big\}$\\\smallskip \bigstrut[t]
regular $n$-gon ($n$ odd) & $n$ & $\big\{-\cos\left(  \frac{2\pi j}{n}\right)
:j=1,\ldots,n\big\}$\\\smallskip
tetrahedron & $4$ & $\big\{  -1,\frac{1}{3}\big\}  $\\\smallskip
octahedron & $6$ & $\big\{  -1,0,1\big\}  $\\\smallskip
cube & $8$ & $\big\{  -1,-\frac{1}{3},\frac{1}{3},1\big\}  $\\\smallskip
cuboctahedron & $12$ & $\big\{  -1,-\frac{1}{2},0,\frac{1}{2},1\big\}
$\\\smallskip
icosahedron & $12$ & $\big\{  -1,-\frac{1}{\sqrt{5}},\frac{1}{\sqrt{5}%
},1\big\}  $\\\smallskip
dodecahedron & $20$ & $\big\{  -1,-\frac{\sqrt{5}}{3},-\frac{1}{3},\frac
{1}{3},\frac{\sqrt{5}}{3},1\big\}  $\\\smallskip
icosidodecahedron & $30$ & $\big\{  -1,-\frac{\tau}{2},-\frac{1}{2},-\frac
{1}{2\tau},0,\frac{1}{2\tau},\frac{1}{2},\frac{\tau}{2},1\big\}  $\\
\noalign{\smallskip}\hline
\end{tabular}
\end{table}
\end{center}
\end{proof}

\begin{remark}
\label{remark}
Let us observe that without any additional calculations we get that the theorem holds true for POVMs represented by regular polygons, tetrahedron, octahedron and icosahedron if the Shannon entropy is replaced by Havrda-Charv\'at-Tsallis $\alpha$-entropy  or R\'enyi $\alpha$-entropy for $\alpha\in(0,2]$. It follows from the fact that the degree of the polynomial $P_v$ interpolating generalised entropy or its increasing function from below (see Sect.~\ref{HerInt}) does not depend on the entropy function. As in all these cases it is at most 2, thus $P$ is constant.
\end{remark}

\begin{remark}
\label{remark2}

In this paper we presented a universal method of determining the global extrema of the entropy of POVM. However, in some cases it is possible to give  proofs that appear to be more elementary.

Let us recall that for tight informationally complete POVMs the sum of squared probabilities of the measurement outcomes (known as the \emph{index of coincidence}) is the same for each initial pure state and equal to $2d/(k(d+1))$. The problem of finding the minimum and maximum of the Shannon entropy under assumption that the index of coincidence is constant has been analyzed in \cite{HarTop01} (some generalizations and related topics can be found also in \cite{BerSan03,Zyc03}). By \cite[Theorem~2.5.]{HarTop01} we get that the minimum is achieved for the probability distribution  $(p,\ldots,p,q,0,\ldots,0) $, where there are $\lfloor k(d+1)/(2d)\rfloor$ probabilities equal to $p$, and both $p$ and $q$ are uniquely determined by the value of the index of coincidence.

One would not suppose this fact to be useful in general setting, since the possible probability distributions of the measurement outcomes for initial pure states form just a $(2d-2)$-dimensional subset of a $(k-2)$-dimensional intersection of a $(k-1)$-sphere and the simplex $\Delta_k$. If $\Pi$ is informationally complete, then $d^2\leq k$. Hence $2d < k$, unless $d=2$ and $k=4$ and so, in general situation, these extremal points need not necessarily belong to this subset. However, this method can be used for the tetrahedral POVM, where $d=2$ and $k=4$.

On the other hand, one can ask whether it is possible to derive a proof of Theorem~\ref{main} for HS-POVMs with octahedral and icosahedral symmetry using the fact that the corresponding Bloch vectors are spherical 3-designs and 5-designs, respectively. The question consists of two problems. The first one is to find the probability distributions that minimize Shannon entropy under assumption that R\'enyi $\alpha$-entropies are fixed for $\alpha=2,3$ and $\alpha=2,3,4,5$, respectively. The second one is whether the obtained extremal probability distribution belongs to the `allowed' set, as the conditions on R\'enyi entropies do not give a complete characterization of this set.

\end{remark}

\begin{remark}
\label{remark3}

Usually, the simplest way to find the entropy minimizers leads through the
majorization technique. However, we shall see that this method fails in general
here and can be useful just in special cases. To show that this is the case, note that if a normalized rank-$1$ POVM $\Pi=\{\Pi_j\}_{j=1}^k$ is tight informationally complete (i.e.\ the set of corresponding pure states is a complex projective 2-design), then for any $\rho\in\mathcal{S}\left(\mathbb{C}^{d}\right)$ the probability distribution of measurement outcomes $(p_1(\rho,\Pi),\ldots,p_k(\rho,\Pi))$ fulfills an additional constraint $p_1(\rho,\Pi)^2+\ldots+p_k(\rho,\Pi)^2=2d/(k(d+1))$. Thus, the set of all possible probability distributions is a $(2d-2)$-dimensional subset of the $(k-1)$-dimensional sphere of radius $2d/(k(d+1))$ intersected with the probability simplex $\Delta_k$. That intersection is a $(k-2)$-dimensional sphere in the affine hyperplane containing $\Delta_k$ that is centered at the uniform distribution and, possibly, cut to fit in the positive hyperoctant (compare Fig.~4 in \cite{AppEriFuc11}). Now, from the fact that the set of probability distributions majorized by a given $P\in\Delta_k$ is a convex hull of its orbit under permutations (see, e.g.\ \cite[Ch.~2.1]{BenZyc06} or \cite[Ch.~1.A]{MarOlkArn11}), it follows that the only probability distributions from the sphere indicated above that majorize (or are majorized by) any probability distribution from the same sphere need to be its permutations.  Hence we deduce that if the distribution of measurement outcomes for one state majorizes that for another one, both distributions must be equivalent, and in particular the measurement entropies at these points are equal. These facts imply that the minimization problem cannot be solved in full generality via majorization.

\end{remark}

\section{Informational power and the average value of relative entropy\label{INFPOW}}

While we know the minimum and maximum values of the relative entropy of some POVMs, it would be worth  taking a look at its average. Surprisingly, the average value of relative entropy over all pure states does not depend on the measurement $\Pi$, but only on the dimension $d$. This can be
proved using (\ref{relent}) and the formula (21) from Jones \cite{Jon91a}.
Namely, we have
\begin{align}
\left\langle \widetilde{H}(\rho,\Pi)\right\rangle _{\rho\in\mathcal{P}(
\mathbb{C}^{d})  }  &  =\int_{\mathcal{P}(  \mathbb{C}^{d})
}\left(  \ln d-\frac{d}{k}\sum_{j=1}^{k}\eta\left(  \operatorname{Tr}%
\left(  \rho\sigma_{j}\right)  \right)  \right)  \textrm{d}m_{FS}\left(  \rho\right)
\label{Euler} \\
&  =\ln d-d\left(  \int_{\mathcal{P}(  \mathbb{C}^{d})
}\eta\left(  \operatorname{Tr}\left(  \rho\sigma_{1}\right)  \right)
\textrm{d}m_{FS}\left(  \rho\right)  \right) \nonumber \\
&  =\ln d-\sum_{j=2}^{d}\frac{1}{j}\rightarrow1-\gamma\text{\quad}\left(
d\rightarrow\infty\right)  \text{,} \nonumber%
\end{align}
where $\gamma\approx0.57722$ is the Euler-Mascheroni constant. This average is
also equal to the maximum value (in dimension $d$) of entropy-like
quantity called \textsl{subentropy}, providing the lower bound for accessible
information \cite{Jozetal94,Datetal13}. Moreover, applying (\ref{ineinfpow}),
Proposition \ref{coninfpow}, and (\ref{Euler}), we get immediately a lower bound
for the informational power of $\Pi$:
\begin{equation}
\ln d-\sum_{j=2}^{d}\frac{1}{j} \ \leq \ W\left( \Pi\right) \text{,}
\end{equation}
provided that condition (1) in Proposition \ref{coninfpow} is fulfilled.
This bound was found, independently, but in general situation, in
\cite{Arnetal14b}.

In particular, the average value of relative entropy is the same for every
HS-POVM $\Pi$ in dimension two and equals $\ln2-1/2\approx0.19315$. It
follows from Theorem \ref{main} and (\ref{relentfor}) that its maximal value,
that is the informational power of~$\Pi$, is given by the formula%

\pagebreak

\begin{equation}
W\left(  \Pi\right)  =\ln2-\frac{2}{|G/G_{v}|}\sum_{\left[  g\right]  \in
G/G_{v}}\eta\left(  \frac{1-gv\cdot v}{2}\right)  \text{,} \smallskip \label{infpowfor}%
\end{equation}
where $G$ is any group acting transitively on the set of Bloch vectors
representing~$\Pi$. Recall that the number of different summands in
(\ref{infpowfor}) is bounded by the number of self-inverse double cosets of
$G_{v}$ plus half of the number of non self-inverse ones.

Applying the above formula to the $n$-gonal POVM we get
\begin{equation}
W\left(  \Pi\right)  =\ln2-\frac{2}%
{n}\sum_{j=1}^{n}\eta\left(  \sin^{2}\frac{\pi j}{n}\right)  \rightarrow
1-\ln2 \, \approx \, 0.30685  \enspace  (n\rightarrow\infty).
\end{equation}
The approximate values of informational power for other HS-POVMs in dimension two can be found in Tab.~\ref{meanentropy}.
It follows from \cite[Corollaries 7-9]{Arn15} that the informational power of tetrahedral, octahedral, and icosahedral POVMs is maximal
among POVMs generated by, respectively, $2$-, $3$- and $5$-designs in dimension two.

\begin{table}[h]
\caption{The approximate values of informational power (maximum relative entropy) for all types of HS-POVMs in dimension two (up to five digits).}
\label{meanentropy}
\centering
\begin{tabular}{lc}
\hline\noalign{\smallskip}
convex hull of the orbit & informational power  \\
\noalign{\smallskip}\hline\noalign{\smallskip}
digon & $0.69315$ \\
regular $n$-gon ($n\rightarrow\infty$) & $0.30685$ \\
tetrahedron & $0.28768$ \\
octahedron & $0.23105$ \\
cube & $0.21576$ \\
cuboctahedron & $0.20273$ \\
icosahedron & $0.20189$ \\
dodecahedron & $0.19686$ \\
icosidodecahedron & $0.19486$ \\
\noalign{\smallskip}\hline\noalign{\smallskip}
average value of relative entropy & $0.19315$ \\
\noalign{\smallskip}\hline
\end{tabular}

\end{table}

Comparing these values to the average value of relative entropy, we see that
the larger is the number of elements in the HS-POVM, the flatter is the graph
of~$\widetilde{H}$; see also Fig.~\ref{F3}, where the graphs in spherical coordinates
are presented.

\begin{figure}[h]
\begin{center}
\includegraphics[scale=0.47]{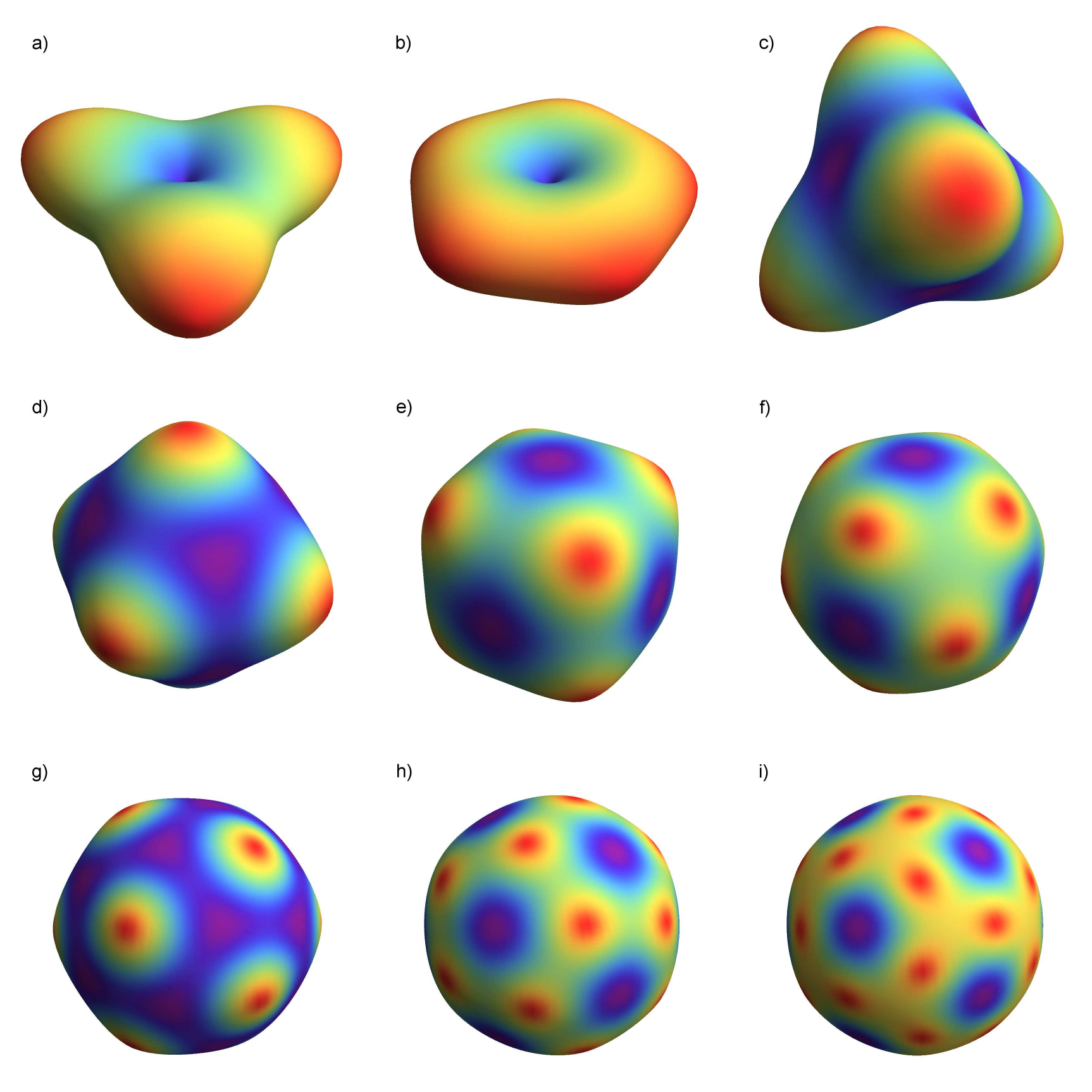}
\end{center}
\caption[The relative entropy of highly symmetric qubit measurements.]{The relative entropy of highly symmetric qubit measurements, where
their Bloch vectors form: a) an equilateral triangle; b) a regular pentagon;
c) a tetrahedron; d) an octahedron; e)~a~cube; f) a cuboctahedron; g) an
icosahedron; h) a dodecahedron; i) an icosidodecahedron. The rainbow-colors
scale that ranges from \emph{red} (maximum) to \emph{purple} (minimum) is used.}%
\label{F3}
\end{figure}

\smallskip

\textit{Acknowledgement}. The authors thank Michele Dall'Arno, S\l awomir Cynk,
\linebreak Piotr Niemiec, and Tomasz Zastawniak for their remarks which improve the presentation
of the paper. Financial support by the Grant No. N N202 090239 of the Polish Ministry
of Science and Higher Education is gratefully acknowledged.

\pagebreak

\begingroup
\renewcommand{\addcontentsline}[3]{}

\endgroup

\end{document}